\documentclass[review,onefignum,onetabnum]{siamonline220329}

\usepackage{lipsum}
\usepackage{amsfonts}
\usepackage{graphicx}
\usepackage{epstopdf}
\usepackage{algorithmic}
\ifpdf
  \DeclareGraphicsExtensions{.eps,.pdf,.png,.jpg}
\else
  \DeclareGraphicsExtensions{.eps}
\fi

\usepackage{enumitem}
\setlist[enumerate]{leftmargin=.5in}
\setlist[itemize]{leftmargin=.5in}

\usepackage{bm}
\usepackage[caption=false]{subfig}
\Crefname{ALC@unique}{Line}{Lines}

\DeclareMathOperator*{\argmin}{argmin}

\newcommand{\op}{\operatorname}

\newcommand{\R}{\mathbb{R}}

\newcommand{\bE}{\mathbb{E}}
\newcommand{\dd}{\mathrm{d}}

\newsiamremark{remark}{Remark}
\newsiamremark{hypothesis}{Hypothesis}
\crefname{hypothesis}{Hypothesis}{Hypotheses}
\newsiamthm{claim}{Claim}
\newsiamthm{assumption}{Assumption} 

\headers{Gradient-enhanced LSM for American options}{J. Yang and G. Li}

\title{Gradient-enhanced sparse Hermite polynomial expansions for pricing and hedging high-dimensional American options\thanks{\funding{JY acknowledges support from the University of Hong Kong via the HKU Presidential PhD Scholar Programme (HKU-PS). GL acknowledges the support from GRF (project number: 17317122) and Early Career Scheme (Project number: 27301921), RGC, Hong Kong.}}}

\author{
Jiefei Yang\thanks{Department of Mathematics, University of Hong Kong, Pokfulam, Hong Kong
  (\email{jiefeiy@connect.hku.hk}).}
\and Guanglian Li\thanks{Department of Mathematics, University of Hong Kong, Pokfulam, Hong Kong
  (\email{lotusli@maths.hku.hk}).}
}

\usepackage{amsopn}

\ifpdf
\hypersetup{
  pdftitle={Gradient-enhanced sparse Hermite polynomial expansions for pricing and hedging high-dimensional American options},
  pdfauthor={J. Yang \and G. Li}
}
\fi

\begin{document}
\nolinenumbers
\maketitle

\begin{abstract}
We propose an efficient and easy-to-implement gradient-enhanced least squares Monte Carlo method for computing price and Greeks (i.e., derivatives of the price function) of high-dimensional American options. It employs the sparse Hermite polynomial expansion as a surrogate model for the continuation value function, and essentially exploits the fast evaluation of gradients. The expansion coefficients are computed by solving a linear least squares problem that is enhanced by gradient information of simulated paths. We analyze the convergence of the proposed method, and establish an error estimate in terms of the best approximation error in the weighted $H^1$ space, the statistical error of solving discrete least squares problems, and the time step size. We present comprehensive numerical experiments to illustrate the performance of the proposed method. The results show that it outperforms the state-of-the-art least squares Monte Carlo method with more accurate price, Greeks, and optimal exercise strategies in high dimensions but with nearly identical computational cost, and it can deliver comparable results with recent neural network-based methods up to dimension 100.
\end{abstract}



\section{Introduction}
The early exercise feature of American or Bermudan options gives holders the right to buy (call) or sell (put) underlying assets before the expiration date, and their accurate numerical calculation is of great practical importance. Meanwhile, the efficient estimation of Greeks (i.e., derivatives of the price function, e.g., delta and gamma) is vital for hedging and risk management, since the theory of option pricing builds on the assumption of the absence of arbitrage. For example, when the asset price is on the rise, the gain in the long position of a call writer's asset may offset the potential loss of the call option.

Nonetheless, the early exercise feature of American options poses significant challenges for computing price and Greeks, especially in high dimensions. One of the most popular methods for high-dimensional American option pricing is the least squares Monte Carlo (LSM) method \cite{longstaff2001valuing, tsitsiklis2001regression}. Computing Greeks in high dimensions is more involved, and further developments with LSM have been proposed \cite{wang2009pricing, bouchard2012monte}.
In this work, building on LSM, we shall develop a simple, fast, and accurate algorithm, termed as gradient-enhanced least squares Monte Carlo (G-LSM) method, cf. \cref{alg: sparse hermite polynomial expansions}, for computing price and Greeks simultaneously at all time steps for dimensions up to 100. The key methodological innovations includes using sparse Hermite polynomial space with a hyperbolic cross index set as the ansatz space for approximating the continuation value functions (CVFs), and incorporating the gradient information for computing the expansion coefficients.

Next we elaborate the two methodological innovations. First, the main obstacle of using gradient information lies in the computational expense: a $d$-variate price has $d$ partial derivatives, which grows quickly with $d$, especially when the derivative evaluation is costly. With the proposed sparse Hermite polynomial ansatz space, the derivatives of polynomial bases can be obtained at almost no extra cost, cf. \cref{eq:derivative_herm} below. This allows greatly reducing the computational cost. Second, although using a polynomial ansatz space for the CVFs as LSM, G-LSM finds  the expansion coefficients via solving a linear least squares problem enhanced by the gradient (and hence the name G-LSM), by minimizing the mean squared error between the approximate and exact value functions at $t_{k+1}$, cf. \cref{eq: empirical risk} below. This differs markedly from LSM, which approximates the conditional expectations by projection and minimizes the mean squared error of approximating CVF at $t_k$. Numerical experiments show that this choice can achieve better accuracy in price, Greeks, and optimal exercise strategies than LSM. Although other classes of basis functions, e.g., Chebyshev, Legendre, and Laguerre polynomials, might be applicable, we propose to adopt sparse Hermite polynomials as the ansatz space for CVF because (1) CVF can be reformulated as a function of independent Brownian motions with its best approximation error ensured by  \cref{lem:best-error-HC} and (2) the gradient of Hermite polynomials can be evaluated very efficiently, which is the key point to enable the algorithm to work in high dimensions.

In G-LSM, the approximation of the terminal condition at $t_{k+1}$ is obtained by discretizing the linear backward stochastic differential equation (BSDE) for the CVF, c.f. \cref{thm: linear bsde}, which was recently innovated in a deep neural network-based method for American option pricing \cite{chen2021deep}. The idea of matching the terminal condition has been widely applied in solving high-dimensional BSDEs with deep neural networks (DNNs) \cite{E2017deepbsde, hure2020deep}. In practice, it involves computing the gradient of the CVF, and in turn that of the basis functions in the ansatz space (in addition to function evaluation). When $d\gg1$, for a complicated ansatz space, evaluating the derivatives at all time steps can be prohibitive. We shall show that the extra cost is nearly negligible for the sparse Hermite polynomial ansatz, and that the overall complexity of G-LSM with $N$ time steps, $M$ sample paths, and $N_b$ basis functions is $\mathcal{O}(NMN_b)$, nearly identical with that for LSM. Numerical results show that the accuracy of G-LSM is competitive with DNN-based methods for dimensions up to $d = 100$.

In theory, the CVF can be formulated as a smooth, high-dimensional function in ${L}^2_{\omega}(\R^d)$ with a Gaussian weight function $\omega(\mathbf{y})$ \cite{yang2023optionpricing}.
This regularity enables the use of normalized and generalized Hermite polynomials, which form an orthonormal basis of $L^2_{\omega}(\R^d)$. Furthermore, drawing on the geometric convergence rate of the hyperbolic cross approximation with Hermite polynomials \cite{luo2018error}, we shall prove the global convergence of G-LSM using BSDE technique, stochastic and Malliavin calculus, and establish an error bound in terms of time step size, statistical error of the Monte Carlo approximation, and the best approximation error in weighted Sobolev space $H^1_{\omega}(\R^d)$, c.f. \cref{thm: global error}. In sum, the algorithmic development of G-LSM, its error analysis and extensive numerical evaluation represent the main contributions of the present work.

Now we situate the present study in existing works. Currently, there are two popular classes of methods to price American options in high dimensions: (i) least-squares Monte Carlo-based (LSM) methods and (ii) DNN-based methods. The LSM method has shown tremendous success for pricing American or Bermudan options with more than one stochastic factors. The original LSM \cite{longstaff2001valuing} uses polynomials to approximate the CVF, and other choices have also been explored, e.g., Gaussian process \cite{ludkovski2018kriging} and DNNs \cite{lapeyre2021neural, becker2020pricing}. Recently, LSM with the hierarchical tensor train technique has been studied in \cite{bayer2023pricing}, which demonstrates the success of polynomial approximation for CVFs in very high dimensions. The proposed G-LSM is a variant of LSM that incorporates gradient information that comes nearly for free. Due to the excellent capability for high-dimensional approximation of DNNs, several methods based on DNNs have been proposed for pricing American or Bermudan options, based on optimal stopping problem (parameterizing the stopping time by DNNs and then maximizing the expected reward \cite{becker2019deep}), free boundary PDEs (parameterizing PDE solutions with DNNs \cite{sirignano2018dgm}), or BSDEs (parameterizing the solution pair of the associated reflected BSDE \cite{el1997reflected} by DNNs \cite{hure2020deep}). Within the framework of BSDEs, Chen and Wan \cite{chen2021deep} suggest approximating the difference of the CVF between adjacent time steps by averaging several trained neural networks, which has a quadratic complexity in the number of time steps. Wang et al. \cite{wang2018deep} extend the deep BSDE method \cite{E2017deepbsde} from European option pricing to Bermudan one, with the loss function being the variance of the initial value, and Gao et al. \cite{gao2023convergence} analyze its convergence. In comparison with DNN-based methods, G-LSM enjoys high efficiency and robustness, which involves only least-squares problems and is easy to implement.

The proposed G-LSM scheme is closely linked to the regression-based Monte Carlo method for solving BSDEs \cite{gobet2005regression}. It involves solving a sequence of least squares problems backward in time, where a linear combination of general basis functions is used to approximate each component of the solution $(Y_{t_k}, Z_{1, t_k}, \dots, Z_{d, t_k})$ to the BSDE.  However, this scheme may become infeasible for high dimensions \cite{hure2020deep}. In sharp contrast, in the $d$-dimensional case, the G-LSM method has a computational complexity $d$ times less that of the approach in \cite{gobet2005regression}, while maintains the stability of least squares regression. See Remark \ref{rmk:complexity} for further discussions on the complexity and stability. The efficiency of the scheme \cite{gobet2005regression} for high-dimensional problems was enhanced by implementing a neural network in the scheme \cite{hure2020deep}, which has demonstrated impressive empirical performance.

The structure of this article is organized as follows. In \cref{sec:pricing-hedging}, we describe the mathematical framework of pricing and hedging high-dimensional American or Bermudan options, and in \cref{sec:sparse-hermite}, we recall several useful properties of generalized Hermite polynomials and approximation with sparse hyperbolic cross index set. Then in \cref{sec:algorithm-cost}, we derive the main algorithm, i.e., gradient-enhanced least squares Monte Carlo (G-LSM) method, and establish its local and global error estimates in \cref{sec:convergence}. In \cref{sec:examples}, we present extensive numerical results including prices, Greeks, optimal stopping time, and computing time. We also present a comparative study with existing methods. Finally, we conclude in \cref{sec:conclusion} with further discussions.

Throughout, bold and plain letters represent multi-variable and scalars, respectively, and the capital and bold letters $\mathbf{S}$, $\Tilde{\mathbf{X}}$ and $\mathbf{W}$ denote random vectors. The notation $\mathbf{a} \cdot \mathbf{b}$ denotes the dot product of two vectors $\mathbf{a}$ and $\mathbf{b}$, $\operatorname{Tr}(A)$ the trace of a square matrix $A$, and $^\top$ vector transpose. The notation
$\nabla_\mathbf{x}f(t, \mathbf{x})$ and  $\operatorname{Hess}_{\mathbf{x}}f(t,\mathbf{x})$ denote respectively the gradient and Hessian of $f$ with respect to  $\mathbf{x}$. For a multi-index $\bm{\alpha} = (\alpha_1,\ldots,\alpha_d)^\top \in \mathbb{N}_0^d$, we denote $|\bm{\alpha}|: = \alpha_1 + \dots + \alpha_d$. For a positive-valued integrable function $\omega: \mathbb{R}^d\to \mathbb{R}^+$, the weighted space $L^2_{\omega}(\mathbb{R}^d)$ is defined by
\begin{align*}
L^2_{\omega}(\mathbb{R}^d):=\{f:\mathbb{R}^d\to \mathbb{R}: \|f\|_{L^2_{\omega}(\mathbb{R}^d)}<\infty\},\quad \text{ with }  \|f\|_{L^2_{\omega}(\mathbb{R}^d)}^2:=\int_{\mathbb{R}^d}f(\mathbf{x})^2\omega(\mathbf{x})\mathrm{d}\mathbf{x}.
\end{align*}
The weighted Sobolev space $H_{\omega}^m(\mathbb{R}^d)$, $m\in \mathbb{N}$, is defined by
$$
H_{\omega}^m(\mathbb{R}^d) := \{f: \R^d \to \R: \|f\|_{H_\omega^m(\R^d)} < \infty\}, \quad \text{with } \|f\|_{H_\omega^m(\R^d)}^2 = \sum_{0\le |\bm{\alpha}|\le m} \left\|\frac{\partial^{\bm{\alpha}} f}{\partial \mathbf{x}^{\bm{\alpha}}}\right\|_{L^2_\omega(\R^d)}^2.
$$
\section{Bermudan option pricing and hedging} \label{sec:pricing-hedging}
Now we describe the valuation framework for American or Bermudan option pricing and hedging.

\subsection{Option pricing and Greeks}
The fair price of American option $v^A(t)$ at time $t\in [0,T]$ is expressed as the solution to the optimal stopping problem in a risk neutral probability space $(\Omega, \mathcal{F}, (\mathcal{F}_t)_{0\le t\le T},\mathbb{Q})$,
\begin{equation*}
    v^A(t) = \sup_{\tau_t \in [t,T]} \mathbb{E}[e^{-r(\tau_t - t)} g(\mathbf{S}_{\tau_t} ) | \mathcal{F}_t],
\end{equation*}
where $\tau_t$ is an $\mathcal{F}_t$-stopping time, $T>0$ is the expiration date, $(\mathbf{S}_t)_{0\le t\le T}$ is a collection of $d$-dimensional price processes, and $g(\mathbf{S}_t) \in L^2(\Omega, \mathcal{F}_t, \mathbb{Q})$ is the payoff depending on the type of the option.

Numerically, the price of Bermudan option is used to approximate the American one. The Bermudan option can be exercised at finite discrete times $0 = t_0 < t_1 < \dots < t_N = T$ with $\Delta t := t_{k+1} - t_k$ for all $k = 0,1,\dots, N-1$. Using dynamic programming principle or Snell envelope theory \cite[Section 1.8.4]{seydel2006tools}, the Bermudan price function ${v}_{t_k}$ at time $t_k$ is given by the following backward induction:
\begin{equation} \label{eq: backward induction}
    \begin{aligned}
        v_{t_N}(\mathbf{s}) &= g_{t_N}(\mathbf{s}), \\
        v_{t_k}(\mathbf{s}) &= \begin{cases}
            g_{t_k}(\mathbf{s}), \quad \text{ if } g_{t_k}(\mathbf{s}) \ge \tilde{c}_{t_k}(\mathbf{s}), \\
            \tilde{c}_{t_k}(\mathbf{s}), \quad \text{ if } g_{t_k}(\mathbf{s}) < \tilde{c}_{t_k}(\mathbf{s}),
        \end{cases} \quad \text{ for } k = N-1:-1:0,
    \end{aligned}
\end{equation}
where $g_{t_k}$ is the discounted payoff (exercise value) function and $\tilde{c}_{t_k}$ is the CVF, defined by
\begin{equation} \label{eq: continuation value function S}
    \tilde{c}_{t_k}(\mathbf{s}) = \mathbb{E}[v_{t_{k+1}}(\mathbf{S}_{t_{k+1}}) | \mathbf{S}_{t_k} = \mathbf{s}].
\end{equation}
The options delta and gamma are defined to be the first and second order derivatives of the price function $v_{t_k}$ with respect to the price of underlying assets:
\begin{equation} \label{eq: defi greeks}
    \begin{aligned}
        \Delta_{t_k} &:= \nabla v_{t_k}(\mathbf{s}) = \left(\frac{\partial v_{t_k}(\mathbf{s})}{\partial s_1}, \dots, \frac{\partial v_{t_k}(\mathbf{s})}{\partial s_d} \right)^\top\quad \mbox{and}\quad
        \Gamma^{ij}_{t_k} &:= \frac{\partial^2 v_{t_k}(\mathbf{s})}{\partial s_j \partial s_i}.
    \end{aligned}
\end{equation}
We consider all exercise and continuation values of Bermudan option discounted to the time $t = 0$. Since the discrete-time approximation for American option pricing is always needed in practical computation, we will make no semantic difference between American and Bermudan options, following \cite[p. 133]{guyon2013nonlinear}.

\subsection{Multi-asset model and transformation} \label{sec:multi model}
One of the most classical models for high-dimensional American option pricing is the multi-asset Black-Scholes model. Under the risk-neutral probability $\mathbb{Q}$, the prices of $d$ underlying assets, $\mathbf{S}_t = \left(S^1_t, \dots, S^d_t \right)^\top$, follow the correlated geometric Brownian motions
\begin{equation} \label{eq: multi asset BS model}
    \mathrm{d}S^i_t = (r - \delta_i) S^i_t\,\mathrm{d}t + \sigma_i S^i_t\,\mathrm{d}\tilde{W}^i_t \quad \text{ with } S^i_0 = s^i_0, \quad i = 1,2,\dots,d,
\end{equation}
where $\tilde{W}^i_t$ are correlated Brownian motions with correlation $\mathbb{E}[\mathrm{d}\tilde{W}^i_t \mathrm{d}\tilde{W}^j_t] = \rho_{ij} \,\mathrm{d}t$, and $r$, $\delta_i$ and $\sigma_i$ are the riskless interest rate, dividend yields,
and volatility parameters, respectively. We
denote the correlation matrix by $P = (\rho_{ij})_{d\times d}$, the volatility matrix by $\Sigma$ (which is a diagonal matrix with
volatility $\sigma_i$ on the diagonal), and write the dividend yields as a vector $\bm{\delta} = [\delta_1, \dots, \delta_d]^\top$. Using the spectral decomposition $\Sigma P \Sigma^\top=Q\Lambda Q^\top$, the rotated log-price $\Tilde{\mathbf{X}}_t := Q^\top \ln(\mathbf{S}_t./\mathbf{s}_0)$ satisfies an independent Gaussian distribution
\begin{align*}
\Tilde{\mathbf{X}}_t\sim
\mathcal{N}\left(Q^\top \left(r - \bm{\delta} - \frac{1}{2}\Sigma^2\mathbf{1}\right)t, \Lambda t\right).
\end{align*}
Let $\bm{\mu} = Q^\top \left(r - \bm{\delta} - \frac{1}{2}\Sigma^2\mathbf{1}\right)$ and $\lambda_i$ be the $i$-th diagonal element of $\Lambda$. This gives a transformation between underlying asset prices $\mathbf{S}_t$ and independent Brownian motions $\mathbf{W}_t = [W^1_t, \dots, W^d_t]^\top$, i.e.,
\begin{equation} \label{eq: trans S and W}
    \mathbf{S}_t = \mathbf{s}_0 \odot \exp\left(Q \left(\bm{\mu}t + \sqrt{\Lambda} \mathbf{W}_t \right) \right),
\end{equation}
where $\odot$ denotes componentwise product.

From \cref{eq: trans S and W} and \cref{eq: continuation value function S}, the CVF $c_{t_k}$ with respect to the independent Brownian motions is
\begin{equation} \label{eq: continuation value W}
    c_{t_k}(\mathbf{w}) = \mathbb{E}[ u_{t_{k+1}}(\mathbf{W}_{t_{k+1}}) | \mathbf{W}_{t_k} = \mathbf{w}],
\end{equation}
where $u_{t_{k+1}}$ represents the value function at time $t_{k+1}$ with respect to the independent Brownian motions. Our aim is to develop a gradient-enhanced least squares Monte Carlo method that can efficiently approximate $c_{t_k}$, and thus provide accurate prices and their derivatives.

\section{Sparse Hermite polynomial expansion and gradient}\label{sec:sparse-hermite}
Sparse polynomial chaos expansion can serve as a surrogate model of unknown stochastic variables with finite second-order moments. The motivations of using sparse Hermite polynomial expansion for pricing and hedging American options are twofold:
\begin{enumerate}
    \item Let $\omega_t$, $t>0$, be the Gaussian density function defined by $\omega_t(\mathbf{y}) := \prod_{j=1}^d \rho_{t}(y_j)$ with $\rho_t(y):= \frac{1}{\sqrt{2\pi t}} \exp(-\frac{y^2}{2t})$. The CVF $c_{t_k}$ satisfies $c_{t_k} \in {L}^2_{\omega_{t_k}}(\R^d)$, since the payoff $g(\mathbf{S}_{t_k})\in L^2(\Omega, \mathcal{F}_{t_k}, \mathbb{Q})$ has a finite second-order moment. Thus, as an orthonormal basis for $L^2_{\omega_{t_k}}(\R^d)$, the set of normalized and generalized Hermite polynomials is a natural choice.
    \item The CVF $c_{t_k} $ has the regularity ${H}^m_{\omega_{t_k}}(\R^d)$ for any positive integer $m$ by \cite[Lemmas 4.2 and 4.3]{yang2023optionpricing}. The smoothness of the function implies efficient polynomial approximation.
\end{enumerate}

Next, we recall one-dimensional normalized and generalized Hermite polynomials $H_n^{(t)}(y)$ and the tensorized $d$-dimensional polynomials $H_{\bm{\alpha}}^{(t)}(\mathbf{y})$, $\bm{\alpha} \in \mathbb{N}_0^d$. For $d = 1$, using the standard probabilist’s Hermite polynomials $H_n(x) = (-1)^n e^{\frac{x^2}{2}} \frac{\dd^n}{\dd x^n}e^{-\frac{x^2}{2}}$, we define the  $n$th-order normalized and generalized Hermite polynomial $H_n^{(t)}$ by
\begin{equation*}
    H_n^{(t)}(y) := \frac{H_n(\frac{y}{\sqrt{t}})}{\sqrt{n!}}, \quad \text{ for } n\in\mathbb{N}_0, t>0, y\in \R.
\end{equation*}
Then $\{H_n^{(t)}\}_{n\in \mathbb{N}_0}$ forms a complete orthonormal basis for ${L}^2_{\rho_{t}}(\R)$:
\begin{equation*}
    \mathbb{E}\left[ H_n^{(t)}(W_t) H_m^{(t)}(W_t) \right] = \int_\R H_n^{(t)}(y) H_m^{(t)}(y) \rho_t(y) \,\mathrm{d}y = \delta_{nm},
\end{equation*}
with $\delta_{nm}$ being the Kronecker delta.
For $d>1$, the tensorized Hermite polynomial with the multi-index $\bm{\alpha}=(\alpha_1,\ldots,\alpha_d)^\top \in \mathbb{N}_0^d$ defined by
\begin{equation*}
    H_{\bm{\alpha}}^{(t)}(\mathbf{y}) := \prod_{j=1}^d H_{\alpha_j}^{(t)}(y_j)
\end{equation*}
forms a complete orthonormal basis for ${L}^2_{\omega_{t}}(\R^d)$, satisfying
\begin{equation*}
    \mathbb{E}\left[ H_{\bm{\alpha}}^{(t)}(\mathbf{W}_t) H_{\bm{\gamma}}^{(t)}(\mathbf{W}_t) \right] = \int_{\R^d} H_{\bm{\alpha}}^{(t)}(\mathbf{y}) H_{\bm{\gamma}}^{(t)}(\mathbf{y})\omega_{t}(\bm{y}) \,\mathrm{d}\mathbf{y} = \delta_{\bm{\alpha}, \bm{\gamma}}, \quad \bm{\alpha}, \bm{\gamma} \in \mathbb{N}_0^d.
\end{equation*}
Here, $\delta_{\bm{\alpha}, \bm{\gamma}} = \prod_{j=1}^d \delta_{\alpha_j,\gamma_j}$ and
$\omega_{t}(\bm{y})=\prod_{j=1}^d \rho_t(y_j) $.

For any fixed multi-index set $I \subset \mathbb{N}_0^d$, the Hermite polynomial ansatz space $P_{I,k}$ is defined by
\begin{align}\label{eq:pkI}
P_{I,k}:=\text{span}\left\{H_{\bm{\alpha}}^{(t_k)}: \bm{\alpha} \in I\right\}.
\end{align}
Then we aim to approximate the CVF $ c_{t_k}$ in $P_{I,k}$ by
\begin{equation*}
    c_{t_k}(\mathbf{W}_{t_k}) \approx \sum_{\bm{\alpha}\in I} \beta_{\bm{\alpha}} H_{\bm{\alpha}}^{(t_k)}(\mathbf{W}_{t_k}).
\end{equation*}
It is well-known that computing polynomial approximations in high dimensions suffers from the notorious curse of dimensionality with tensor product-type multi-index sets. Fortunately, the smoothness of $c_{t_k}$ implies a fast decay of the coefficients in the polynomial expansion. The large coefficients usually occur in a lower multi-index set $I$, that is, if $\bm{\alpha}\in I$ and $\bm{\gamma} \le \bm{\alpha}$, then $\bm{\gamma}\in I$ \cite[Section 1.5]{adcock2022sparse}. To circumvent the curse of dimensionality,  the decay property of the expansion coefficients can be exploited and a hyperbolic cross sparse index set can be used to construct an approximation. The hyperbolic cross multi-index set $I$ with maximum order $p\in \mathbb{N}$ is defined by
\begin{equation} \label{eq:HC-defi}
    I := \Big\{\bm{\alpha} = (\alpha_j)_{j=1}^d \in \mathbb{N}_0^d: \prod_{j=1}^d \max(\alpha_j, 1) \le p\Big\},
\end{equation}
which has a cardinality $\mathcal{O}(p(\ln p)^{d-1})$. The best approximation error of the sparse Hermite polynomial approximation with hyperbolic cross index set was analyzed in \cite{luo2018error}; see \cref{sec:convergence} for details.

Now, we introduce a property of derivatives of Hermite polynomials, which plays an important role in reducing the computational cost. The first-order derivative of the one-dimensional normalized and generalized Hermite polynomial $H_n^{(t)}(y)$ satisfies
\begin{equation*}
    \frac{\mathrm{d}}{\mathrm{d}y}\left(H_n^{(t)}(y)\right) = \sqrt{\frac{n}{t}} H_{n-1}^{(t)}(y).
\end{equation*}
For the $d$-dimensional Hermite polynomial, we have
\begin{equation} \label{eq:derivative_herm}
    \frac{\partial}{\partial y_j}\left( H_{\bm{\alpha}}^{(t)}(\mathbf{y}) \right) = \sqrt{\frac{\alpha_j}{t}} H_{\bm{\alpha} - \mathbf{e}_j}^{(t)}(\mathbf{y}),
\end{equation}
where $\mathbf{e}_j$ is the $j$-th canonical basis vector. This implies that for a lower multi-index set $I \subset \mathbb{N}_0^d$, we have $\bm{\alpha} - \mathbf{e}_j \in I$ if $\bm{\alpha}\in I$. Thus, once the evaluations of $H_{\bm{\alpha}}^{(t)}(\mathbf{y})$ for $\bm{\alpha}\in I$ are available, the gradients of $H_{\bm{\alpha}}^{(t)}(\mathbf{y})$ for $\bm{\alpha}\in I$ can be evaluated cheaply.

\section{Algorithm and complexity} \label{sec:algorithm-cost}
Now we derive the main methodology to approximate the CVF $c_{t_k}$ by matching values of $u_{t_{k+1}}$, and analyze its computational complexity.
Below we abbreviate the notations $c_{t_k}$, $u_{t_k}$ and $\mathbf{W}_{t_k}$ to $c_k$, $u_k$ and $\mathbf{W}_k$, etc, for $k=0,1,\dots, N$.

\subsection{Gradient-enhanced Least Squares (G-LS)}
First we derive a linear backward stochastic differential equation (BSDE) for the CVF $ c_k$. It is worth noting that the corresponding BSDE does not contain a generator, and the option price is solely a function of Brownian motion rather than the underlying asset. This fact is crucial for developing the G-LSM using sparse Hermite polynomials.
\begin{theorem} \label{thm: linear bsde}
The CVF $c_k(\mathbf{W}_k)$ satisfies the linear BSDE
    \begin{equation*}
        c_k(\mathbf{W}_k) = u_{k+1}(\mathbf{W}_{k+1}) - \int_{t_k}^{t_{k+1}} \nabla_{\mathbf{w}}f(t, \mathbf{W}_t) \cdot \mathrm{d}\mathbf{W}_t,
    \end{equation*} 
where $f(t, \mathbf{w}), t\in [t_k, t_{k+1}]$, is defined by
$f(t, \mathbf{w}) := \mathbb{E}[ u_{k+1}(\mathbf{W}_{k+1}) | \mathbf{W}_t = \mathbf{w} ].$
\end{theorem}
\begin{proof}
This is a direct consequence of martingale representation theorem \cite[Theorem 2.5.2]{zhang2017backward}. For the sake of completeness, we provide a brief proof. 
By the Feynman-Kac formula, $f(t, \mathbf{w})$ satisfies a $d$-dimensional parabolic PDE subject to a terminal condition
\begin{equation}\label{eq:para-terminal}
\left\{
\begin{aligned}
\frac{\partial f}{\partial t}(t, \mathbf{w}) + \frac{1}{2}\operatorname{Tr}\left(\operatorname{Hess}_{\mathbf{w}}f(t,\mathbf{w}) \right) &= 0,\quad  t\in [t_k, t_{k+1}), \\
f(t_{k+1}, \mathbf{w}) &= u_{t_{k+1}}(\mathbf{w}).
\end{aligned}
\right.
\end{equation}
Let $Y_t := f(t, \mathbf{W}_t)$. By It\^{o}'s formula, we have
    \begin{equation*}
        \dd Y_t = \left( \frac{\partial f}{\partial t}(t, \mathbf{W}_t) + \frac{1}{2}\operatorname{Tr}\left(\operatorname{Hess}_{\mathbf{w}}f(t,\mathbf{W}_t) \right) \right) \,\mathrm{d}t + \nabla_{\mathbf{w}}f(t, \mathbf{W}_t) \cdot \mathrm{d}\mathbf{W}_t.
    \end{equation*}
In view of \cref{eq:para-terminal}, the drift term vanishes. After taking the stochastic integral and using the terminal condition in \cref{eq:para-terminal}, we obtain the desired assertion.
\end{proof}

\begin{remark}
The Feynman-Kac formula in the classical book \cite{oksendal2013stochastic} requires $C^2$ regularity of the terminal condition, which is not satisfied by the value function. The regularity requirement can be relaxed by considering the generalized result of nonlinear Feynman-Kac formula in \cite[p. 105]{zhang2017backward}, which only requires the terminal condition to be uniformly Lipschitz continuous.
\end{remark}

Next, we construct an approximation of the CVF ${c}_{k}$ by matching the terminal condition over the interval $[t_k,t_{k+1}]$. By \cref{thm: linear bsde}, the terminal value over $[t_k,t_{k+1}]$ can be approximated by the Euler discretization:
\begin{equation} \label{eq: approximation equation}
     \bar{u}_{k+1}(\mathbf{W}_{k+1}) =
     {c}^{{\rm CLS}}_{k}(\mathbf{W}_k) + \nabla {c}^{\op{CLS}}_{k}(\mathbf{W}_k) \cdot \Delta\mathbf{W}_k,
\end{equation}
where $\Delta\mathbf{W}_k:=\mathbf{W}_{k+1} - \mathbf{W}_k$ is the Brownian increment and ${c}^{\op{CLS}}_{k}$ denotes an approximation to the CVF ${c}_{k}$. Let $\hat{u}_{k+1}$ be the value function computed in the last time step. Then using $P_{I,k}$ defined in \cref{eq:pkI} and \cref{eq:HC-defi} as the ansatz space for ${c}^{\op{CLS}}_{k}$, we solve for ${c}^{\op{CLS}}_{k}$ using the least squares regression:
\begin{equation} \label{eq: expected risk}
    {c}_{k}^{\op{CLS}}(\mathbf{W}_k) = \argmin_{\psi \in P_{I,k}} E_k(\psi),
\end{equation}
with $E_k(\cdot): P_{I,k}\to \mathbb{R}^+$ being the quadratic loss (i.e., mean squared error) defined by
\begin{align*}
E_k(\psi):=\mathbb{E}\left[\left( \hat{u}_{k+1}(\mathbf{W}_{k+1}) - \psi(\mathbf{W}_k) - \nabla \psi(\mathbf{W}_k) \cdot \Delta \mathbf{W}_k \right)^2 \right].
\end{align*}
Finally, the numerical value $\hat{u}_k$ at time $t_k$ is updated to be the discounted exercise or continuation value at $t_k$:
\begin{align}\label{eq:num-val}
\hat{u}_{k}(\mathbf{W}_k)=\begin{cases}
    g_k(\mathbf{S}_k), \quad &\text{ if exercise}, \\
    c_k^{\op{CLS}}(\mathbf{W}_k), \quad &\text{ if continue}.
\end{cases}
\end{align}
Since the option is only profitable when exercised in the in-the-money region $\Omega_{\op{ITM}} = \{\mathbf{s}\in \R^d: g_k(\mathbf{s})>0\}$, we make the decision of exercising the option when $c_k^{\op{CLS}}(\mathbf{W}_k) < g_k(\mathbf{S}_k)$ and $\mathbf{S}_k \in \Omega_{\op{ITM}}$; otherwise, the option will be continued.

\subsection{Gradient-enhanced Least Squares Monte Carlo (G-LSM)} \label{subsec:glsm}
Now we derive the methodology based on the Monte Carlo method to solve \cref{eq: expected risk} numerically and analyze its computational complexity. Let $N_b=|I| $ and let $\{\phi_n^{k}(\mathbf{W}_k)\}_{n=1}^{N_b}$ be the set of Hermite polynomials in a scalar-indexed form. Then any function $\psi\in P_{I,k}$ can be expressed as
\begin{align*}
\psi(\mathbf{W}_k) = \sum_{n=1}^{N_b} \beta_n \phi_n^{k}(\mathbf{W}_k),
\end{align*}
with its gradient given by
\begin{equation} \label{eq: gradient of approx cv}
    \nabla \psi(\mathbf{W}_k) = \sum_{n=1}^{N_b} \beta_n \nabla \phi_n^{k}(\mathbf{W}_k).
\end{equation}
In practice, the continuous least squares problem \cref{eq: expected risk} is solved by minimizing its Monte Carlo approximation:
\begin{equation} \label{eq: empirical risk}
    \hat{c}_k := \argmin_{\psi \in P_{I,k}}
    \frac{1}{M}\sum_{m=1}^M
    \big(\hat{u}_{k+1}(\mathbf{W}^m_{k+1}) -
    \psi(\mathbf{W}^m_k) - \nabla \psi(\mathbf{W}^m_k) \cdot  \Delta\mathbf{W}^m_k
    \big)^2,
\end{equation}
where $\{\mathbf{W}^m_k\}_{m=1}^M$ are $M$ independent paths of the Gaussian random process $\mathbf{W}_k$ and $\Delta\mathbf{W}^m_k:=\mathbf{W}^m_{k+1} - \mathbf{W}^m_k$ is the $m$th path of the increment $\Delta \mathbf{W}_k$. Let $A_k\in\mathbb{R}^{M\times N_b}$, $\bm{\beta}_k\in\mathbb{R}^{ N_b}$ and $\mathbf{\hat{u}}_{k+1}\in\mathbb{R}^{ M}$ with their components defined by
\begin{equation} \label{eq:defi-linear-system}
    \begin{aligned}
        (A_k)_{m,n} &= \phi_n^{k}(\mathbf{W}^m_k) + \nabla \phi_n^{k}(\mathbf{W}^m_k) \cdot \Delta\mathbf{W}^m_k , \\
        (\bm{\beta}_k)_n &= \beta_n, \quad 
        (\mathbf{\hat{u}}_{k+1})_m = \hat{u}_{k+1}(\mathbf{W}^m_{k+1})
    \end{aligned}
\end{equation}
for $m = 1,\dots, M$, $n = 1,\dots, N_b$. 

Then finding the optimal polynomial in \cref{eq: empirical risk} amounts to solving the classical least squares problem
\begin{equation}\label{eq:ls-linearSystem}
    \bm{\beta}_k = \argmin_{\bm{\beta}} \|A_k \bm{\beta} - \mathbf{\hat{u}}_{k+1}\|_2^2 = (A_k^\top A_k)^{-1}A_k^\top \mathbf{\hat{u}}_{k+1}.
\end{equation}
The proposed algorithm is summarized in \cref{alg: sparse hermite polynomial expansions}.

\begin{remark}
Note that the well-known least squares Monte Carlo (LSM) method  \cite{longstaff2001valuing, tsitsiklis2001regression} also approximates the CVF with a finite number of basis functions. The orthonormal Hermite polynomials is one possible choice. However, LSM computes the coefficients by projecting the value, $u_{t_{k+1}}(\mathbf{W}_{t_{k+1}})$ in \cref{eq: continuation value W}, onto a finite-dimensional space spanned by the basis functions due to the projection nature of conditional expectation:
\begin{equation*}
    \hat{c}_k^{LSM} := \argmin_{\psi \in P_{I,k}}
    \frac{1}{M}\sum_{m=1}^M
    \big(\hat{u}_{k+1}(\mathbf{W}^m_{k+1}) -
    \psi(\mathbf{W}^m_k) \big)^2.
\end{equation*}
In contrast, G-LSM computes the coefficients by matching the approximate and exact value of $u_{t_{k+1}}(\mathbf{W}_{t_{k+1}})$, see \cref{eq: empirical risk}. Their numerical performance will be compared in  \cref{sec:examples}.
\end{remark}

\begin{algorithm}[htbp]
\caption{G-LSM}
\label{alg: sparse hermite polynomial expansions}
\begin{algorithmic}[1]
\REQUIRE Market parameters: $ \mathbf{S}_0, r, \delta_i, \sigma_i, P $ \par
        Option parameters: payoff function $g(\mathbf{s})$ \par
        Algorithm parameters: $N, M, p$
\ENSURE Option price $v_0$
\STATE Compute the hyperbolic cross multi-index set $I$ with maximum order $p$ 
\STATE Generate $M$ sample paths
\STATE Initialize values $\mathbf{\hat{u}}_N = (e^{-rT}g(\mathbf{S}_N^m))_{m=1}^M$
\STATE Initialize stopping times $\tau_\star = T$
\FOR {$k = N-1:-1:1$}
    \STATE $\Phi_k = $ basis matrix$(\mathbf{W}_k, I)$ with $(\Phi_k)_{m,n} = \phi_n^k(\mathbf{W}_k^m)$
    \STATE Compute matrix $A_k$ with \cref{alg: compute matrix A}
    \STATE Solve system of linear equations: $A_k\bm{\beta}_k = \mathbf{u}_{k+1}$
    \STATE Update $(\mathbf{\hat{u}}_{k})_m = \begin{cases}
        e^{-rk\Delta t} g(\mathbf{S}^m_k), &\text{if exercise} \\
        (\Phi_k \bm{\beta}_k)_m, &\text{if continue}
    \end{cases}
    \text{ and }
    \tau_\star = \begin{cases}
        k\Delta t,  &\text{if exercise} \\
        \tau_\star, &\text{if continue}
    \end{cases}$
\ENDFOR
\STATE $v_0 = \max\left\{\frac{1}{M} \sum_{m=1}^M e^{-r\tau_\star^m} g(\mathbf{S}^m_{\tau_\star^m}),g(\mathbf{S}_0)\right\}$
\end{algorithmic}
\end{algorithm}

For the efficient computation of the matrix $A_k$, using \cref{eq:derivative_herm}, i.e.,
\begin{equation} \label{eq:basis-derivative}
    \frac{\partial \phi_n^{k}}{\partial w_j}(\mathbf{W}^m_k) = \sqrt{\frac{\alpha_j}{t_k}} \phi_{n'}^{k}(\mathbf{W}^m_k) \quad \text{ for } \phi_n^{k} = H_{\bm{\alpha}}^{(t_k)}, \phi_{n'}^{k} = H_{\bm{\alpha} - \mathbf{e}_j}^{(t_k)},
\end{equation}
we can compute the matrix $A_k$ from the basis matrix $\Phi_k$ and the increment $\Delta\mathbf{W}_k$. This routine is summarized in \cref{alg: compute matrix A}, where $(A_k)_{n}$ represents the $n$-th column of the matrix $A_k$.

Finally, we analyze the computational complexity of \cref{alg: sparse hermite polynomial expansions}. We employ backward induction \cref{eq: backward induction} with $N$ time steps to price American options, which has a complexity $\mathcal{O}(N)$. For each fixed $t_k$, the computation consists of steps 6-9 in \cref{alg: sparse hermite polynomial expansions}. Step 6 has a complexity $\mathcal{O}(MN_b)$, when using $M$ samples and $N_b$ basis functions. Step 7 is detailed in \cref{alg: compute matrix A}, which has a linear complexity in the number of nonzero $\alpha_j$ for all $\bm{\alpha} \in I$ and $j=1,\dots,d$, 
\begin{align*}
\|I\|_0:=\sum_{\bm{\alpha}\in I}\#\left\{j=1,\cdots,d: \alpha_j\neq 0\right\}. 
\end{align*}
Note that $\|I\|_0$ can be calculated by $(N_b - N_{b,p-1})d$ for maximum order $p$, where $N_{b,p-1}$ represents the number of basis functions with maximum order $p-1$. \Cref{fig: hc_nnz} shows that $\|I\|_0$ scales nearly linearly with respect to $N_b$. Thus, the complexity of \cref{alg: compute matrix A} is linear in $N_b$. Step 8 involves solving a linear system of size $M\times N_b$. Using an iterative solver, the cost is $\mathcal{O}(MN_b)$ if the matrix $A$ is well-conditioned. Step 9 involves matrix-vector multiplication, which has a complexity $\mathcal{O}(MN_b)$. Hence, with a fixed sample size $M$ and number $N$ of time steps, the total cost of the proposed G-LSM is $\mathcal{O}(NMN_b)$. We will numerically verify the complexity analysis in \cref{subsec:ex3}.

\begin{algorithm}[htbp]
\caption{Compute the matrix $A_k$ using $I, \Phi_k$ and $\Delta\mathbf{W}_k$}
\label{alg: compute matrix A}
\begin{algorithmic}[1]
\STATE Initialize $A_k = \Phi_k$ 
\FOR{$j=1:d$}
    \FOR{$\bm{\alpha} \in I$}
        \IF{$\alpha_j \ge 1$}
        \STATE $(A_k)_{n} = (A_k)_{n} + (\Delta W_k)_j \odot (\Phi_k)_{n'} \sqrt{\frac{\alpha_j}{t_k}}$.
        \ENDIF
    \ENDFOR
\ENDFOR
\end{algorithmic}
\end{algorithm}

\begin{figure}[htbp]
    \centering
    \begin{tabular}{ccc}
    \includegraphics[width = .3\textwidth]{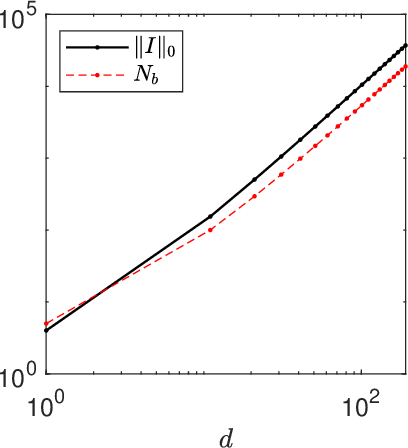}
    &
    \includegraphics[width = .3\textwidth]{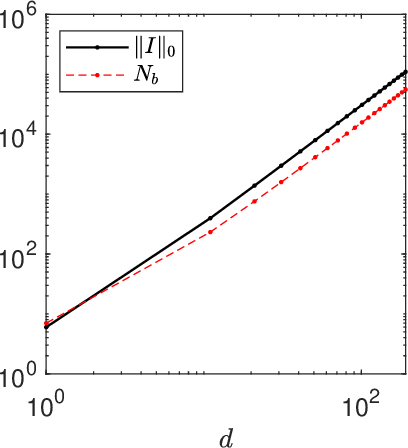}
    &
    \includegraphics[width = .3\textwidth]{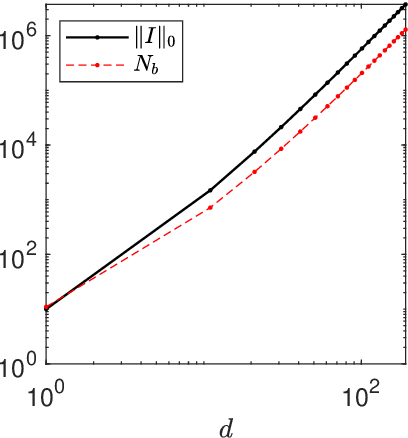}\\
    (a) $p=4$ & (b) $p=6$ & (c) $p=10$
    \end{tabular}
    \caption{$\|I\|_0$ scales linearly with respect to $N_b$.}
    \label{fig: hc_nnz}
\end{figure}

\subsection{Stability of the system of linear equations (\ref*{eq:ls-linearSystem})}
Next, we derive the stability of the system of linear equations \eqref{eq:ls-linearSystem} by estimating the condition number $\mathrm{cond}(A_k^\top A_k)$ (under the standard Euclidean norm) of the matrix $A_k^\top A_k$ defined in \eqref{eq:defi-linear-system} when the number $M$ of simulation paths is large. This result underpins the efficiency of the standard conjugate gradient method \cite{Hestenes:1952} for solving the least-squares problem (the associated normal equation, to be more precise).

\begin{theorem}\label{thm:cond}
Let $A_k\in\mathbb{R}^{M\times N_b}$ be defined in \eqref{eq:defi-linear-system}. Then there holds
\begin{align*}
\lim_{M\to\infty}\operatorname{cond}(A_k^\top A_k)= 1+\frac{p+d-1}{k} \quad\text{almost\; surely.}
\end{align*}
\end{theorem}
\begin{proof}
By definition, the $(i,j)$ entry of matrix $A_k^\top A_k$ is given by 
\begin{equation} \label{eq:entry-matrix}
    (A_k^\top A_k)_{i,j} = \sum_{m=1}^M \left( \phi_i^k(\mathbf{W}_k^m) + \nabla \phi_i^k(\mathbf{W}_k^m)\cdot \Delta \mathbf{W}_k^m\right) \left( \phi_j^k(\mathbf{W}_k^m) + \nabla \phi_j^k(\mathbf{W}_k^m)\cdot \Delta \mathbf{W}_k^m\right).
\end{equation}
By the vanilla Monte Carlo approximation, we obtain 
\begin{align}
        &\quad \lim_{M\to\infty} \frac{1}{M}\sum_{m=1}^M \left( \phi_i^k(\mathbf{W}_k^m) + \nabla \phi_i^k(\mathbf{W}_k^m)\cdot \Delta \mathbf{W}_k^m\right) \left( \phi_j^k(\mathbf{W}_k^m) + \nabla \phi_j^k(\mathbf{W}_k^m)\cdot \Delta \mathbf{W}_k^m\right)\nonumber \\
        &= \mathbb{E}\left[ \phi_i^k(\mathbf{W}_k)\phi_j^k(\mathbf{W}_k) + \nabla \phi_i^k(\mathbf{W}_k)\cdot \nabla\phi_j^k(\mathbf{W}_k) (\Delta \mathbf{W}_k)^2 \right]\nonumber \\
        &\quad + \mathbb{E}\left[ \left(\phi_i^k(\mathbf{W}_k) \nabla\phi_j^k(\mathbf{W}_k) + \phi_j^k(\mathbf{W}_k) \nabla\phi_i^k(\mathbf{W}_k) \right)\cdot \Delta \mathbf{W}_k \right] a.s. \nonumber \\
        &=  \mathbb{E}\left[ \phi_i^k(\mathbf{W}_k)\phi_j^k(\mathbf{W}_k) + \nabla \phi_i^k(\mathbf{W}_k)\cdot\nabla\phi_j^k(\mathbf{W}_k) (\Delta \mathbf{W}_k)^2 \right],\label{eq:entry-form}
\end{align}
where the last equality is due to the vanishing mean of Brownian motion increment $\mathbb{E}[\Delta \mathbf{W}_k] = 0$.

Furthermore, since $\phi_i^k$ and $\phi_j^k$ are generalized and normalized Hermite polynomials, then $\mathbb{E}[\phi_i^k(\mathbf{W}_k)\phi_j^k(\mathbf{W}_k)] = \delta_{ij}$. Then the identity $\mathbb{E}[(\Delta \mathbf{W}_k)^2] = \Delta t$ and \eqref{eq:basis-derivative} lead to 
\begin{equation*}
    \mathbb{E}[\nabla \phi_i^k(\mathbf{W}_k)\cdot\nabla\phi_j^k(\mathbf{W}_k) (\Delta \mathbf{W}_k)^2] = \bE\left[\sum_{\ell = 1}^d \partial_\ell \phi_i^k(\mathbf{W}_k)\partial_\ell \phi_j^k(\mathbf{W}_k)\right]\Delta t = 0,\quad \text{if }i\ne j.
\end{equation*}
This, together with \eqref{eq:entry-matrix} and \eqref{eq:entry-form}, implies that if $i\ne j$, then $\lim_{M\to\infty}(A_k^\top A_k)_{i,j} = 0$ a.s., i.e., the off-diagonal entries are nearly vanishing when $M\to \infty$.
To estimate the diagonal entries, let $\phi_i^k(\mathbf{W_k}) = H_{\bm{\alpha}}^{(t_k)}$ for some multi-index $\bm{\alpha}$. Using the first order derivative formula \eqref{eq:basis-derivative} and $t_k = k\Delta t$, we derive 
\begin{equation} \label{eq:element-AA}
    \begin{aligned}
        \bE\left[\sum_{\ell = 1}^d \partial_\ell \phi_i^k(\mathbf{W}_k)\partial_\ell \phi_i^k(\mathbf{W}_k)\right]\Delta t &= \Delta t\sum_{\ell = 1}^d \bE\left[ \left(\partial_\ell H_{\bm{\alpha}}^{(t_k)}(\mathbf{W}_k)\right)^2 \right] \\
        &=\Delta t\sum_{\ell = 1}^d \frac{\alpha_\ell}{t_k}\bE\left[ \left( H_{\bm{\alpha} - \mathbf{e}_\ell}^{(t_k)}(\mathbf{W}_k) \right)^2 \right] = \Delta t\frac{|\bm{\alpha}|}{t_k} = \frac{|\bm{\alpha}|}{k}.
    \end{aligned}
\end{equation}
By combining \eqref{eq:element-AA} with \eqref{eq:entry-form}, we obtain an estimate of the diagonal entry, \[\lim_{M\to\infty}\frac{1}{M}(A_k^\top A_k)_{i,i}= (1 + \frac{|\bm{\alpha}|}{k})\; a.s.
\] Here, the involved approximation error is only due to the Monte Carlo method. Thus, the matrix $A_k^\top A_k$ is approximately a diagonal matrix when the sample size $M$ is large. Hence, the definition of condition number implies
\begin{equation*}
     \lim_{M\to\infty}\operatorname{cond}(A_k^\top A_k)=\left(1 + \frac{\max_{\bm{\alpha \in I}}|\bm{\alpha}|}{k}\right) \Big/ \left(1 + \frac{\min_{\bm{\alpha \in I}}|\bm{\alpha}|}{k}\right) = 1 + \frac{p+d-1}{k} a.s.
\end{equation*}
Here, $I$ is the hyperbolic index set defined in \eqref{eq:HC-defi}. This proves the assertion.
\end{proof}

\begin{remark}\label{rmk:complexity}
We give a detailed comparison of the proposed G-LSM method with the numerical scheme in \cite{gobet2005regression} for solving BSDEs in terms of computational complexity. The scheme in \cite{gobet2005regression} solves the least squares problem 
    \begin{equation} \label{eq:ls-GLW05}
        \min_{\alpha_0, \alpha_1, \dots, \alpha_d} \frac{1}{M}\sum_{m=1}^M \left( Y_{t_{k+1}}^{N, I, I, M, m} - \alpha_0\cdot p_{0,k}^m - \sum_{l=1}^d \alpha_l \cdot p_{l,k}^m \Delta W_{l,k}^m \right)^2,
    \end{equation}
where each component of the BSDE solution $(Y_{t_k}, Z_{1, t_k}, \dots, Z_{d, t_k})$ is approximated by the linear combination of distinct basis functions, i.e., $\alpha_l\cdot p_{l,k}$. Suppose each approximation uses $N_b$ basis functions. Then $\alpha_l$ is a vector of length $N_b$ for each $l\in \{0, 1,\dots, d\}$. In total, the least squares problem \eqref{eq:ls-GLW05} involves $(d+1)\times N_b$ unknowns for a $d$-dimensional problem.  In contrast, the least squares problem in \eqref{eq: empirical risk} is equivalent to 
    \begin{equation} \label{eq:ls-glsm}
        \min_{\beta} \frac{1}{M}\sum_{m=1}^M \left( \hat{u}_{k+1}(\mathbf{W}_{k+1}^m) - \beta \cdot (A_k)_{m,\cdot} \right)^2,
    \end{equation}
    where only the CVF is approximated by the linear combination of sparse Hermite basis functions, i.e, the unknown $\beta \cdot \phi^k$ with $\beta\in \mathbb{R}^{N_b}$ and $\phi^k$ being the vector of basis functions. Here, $(A_k)_{m,\cdot}$ is the $m$-th row of the matrix $A_k$ in the linear system, which can be efficiently calculated by Algorithm 2. To the best of our knowledge, this idea is new in literature. We can see that the least squares problem \eqref{eq:ls-glsm} includes $N_b$ unknowns for a $d$-dimensional problem. Compared with $(d+1)\times N_b$ unknowns in the system \eqref{eq:ls-GLW05}, the least squares problem \eqref{eq:ls-glsm} incurs much lower computational complexity. Furthermore, the empirical regression matrix for the problem \eqref{eq:ls-GLW05} may not be invertible when the sample size becomes large as mentioned in \cite[page 2178]{gobet2005regression}, while the regression matrix for \eqref{eq:ls-glsm} is well-conditioned for large sample size as proved in Theorem \ref{thm:cond}.  In summary, in comparison with the scheme in \cite{gobet2005regression}, the proposed G-LSM method has $d$ times less computational complexity for  a $d$-dimensional problem, and  enjoys guaranteed stability when $M\to \infty$.
\end{remark}

\subsection{Computing deltas}
Now we present the approximation of deltas for constructing hedging strategies based on \cref{alg: sparse hermite polynomial expansions}. In the backward induction loop, \cref{eq: backward induction} and \cref{eq: defi greeks} imply the deltas at time $t_k$ is given by
\begin{equation*}
    \frac{\partial v_k(\mathbf{s})}{\partial s_j} = \left\{
        \begin{aligned}
            \frac{\partial g_k(\mathbf{s})}{\partial s_j}, \text{ if exercise}, \\
            \frac{\partial {\tilde{c}}_k(\mathbf{s})}{\partial s_j}, \text{ if continue}, 
        \end{aligned}
    \right.\quad j = 1,\dots, d.
\end{equation*}
In the continuation region, $\tilde{c}_k(\mathbf{s}) = c_k(\mathbf{w})$ with $\mathbf{s} = \mathbf{s}_0 \odot \exp(Q(\mu t_k + \sqrt{\Lambda}\mathbf{w}))$.
Once we have obtained the coefficients $\bm{\beta}_k$ of the expansion, the gradient $\nabla c_k(\mathbf{w})$ of the continuation value $c_k(\mathbf{w})$ is approximated by $\nabla \hat{c}_k(\mathbf{w})$ in \cref{eq: gradient of approx cv}. Hence, we calculate the deltas in the continuation region by
\begin{equation*}
    \frac{\partial {\tilde{c}}_k(\mathbf{s})}{\partial s_j} \approx \sum_{i=1}^d \frac{\partial \hat{c}_k(\mathbf{w})}{\partial w_i}\frac{\partial w_i}{\partial s_j}, \quad k = 1,\dots, {N-1}.
\end{equation*}
The gammas $\Gamma_{t_k}^{ij}$ can be approximated in a similar way.

Particularly, the Greeks at time $t=0$ involves solving an additional linear system, which arises from minimizing the $L_{\omega_{t_1}}^2(\R^d)$ error of approximating $u_{1}(\mathbf{W}_1)$. We employ a slightly different ansatz from the case of $t_k>0$, since the expansions in $H_{\bm{\alpha}}^{(0)}$ is not applicable. Instead, we take 
\begin{equation*}
    \hat{c}_0(\mathbf{w}) = \sum_{\bm{\alpha}\in I} \beta_{\bm{\alpha}} H_{\bm{\alpha}}^{(\Delta t)}(\mathbf{w}).
\end{equation*}
Since $c_0(\mathbf{W}_0)$ is a deterministic value given $\mathbf{W}_0 = \mathbf{0}$, we consider the Euler approximation
\begin{equation*}
    u_1(\mathbf{W}_1) \approx \hat{c}_0(\mathbf{0}) + \nabla \hat{c}_0(\mathbf{0}) \cdot \mathbf{W}_1.
\end{equation*}
Similar to \cref{eq: empirical risk}, the coefficients $\beta_{\bm{\alpha}}$ can be approximated by solving an $M\times N_b$ linear system. After that, the delta at $t=0$ is approximated by
\begin{equation*}
    \frac{\partial {\tilde{c}}_0(\mathbf{s}_0)}{\partial s_j} \approx \sum_{i=1}^d \frac{\partial \hat{c}_0(\mathbf{0})}{\partial w_i}\frac{\partial w_i}{\partial s_j}(\mathbf{s}_0),\quad \mbox{with }\mathbf{s} = \mathbf{s}_0 \odot  \exp(Q\sqrt{\Lambda}\mathbf{w}).
\end{equation*}

\section{Convergence analysis} \label{sec:convergence}
Now we analyze the convergence of \cref{alg: sparse hermite polynomial expansions}. 
We restrict the theoretical analysis to the multi-asset Black-Scholes model presented in \cref{sec:multi model} since the corresponding CVF can be reformulated as a function of independent Brownian motions and the associated best approximation error in the sparse Hermite polynomial ansatz space $P_{I,k}$ can be rigorously established, cf.  \cref{lem:best-error-HC}. In practice, \cref{alg: sparse hermite polynomial expansions} can be applied to general models driven by It\^{o} diffusions.  See  \cref{subsec:heston} for one numerical example related to the Heston model.
We assume the Lipschitz continuity of the discounted payoff function $g_k(\cdot)$.
\begin{assumption}\label{assu:lip}
    The discounted payoff function $g_k(\cdot)$ is $L$-Lipschitz continuous: for  $k = 0,1,\dots,N$
    \begin{equation} \label{eq: Lipschitz assum}
        |g_k(\mathbf{s}) - g_k(\mathbf{s}')| \le L\|\mathbf{s} - \mathbf{s}'\|, \quad \forall \mathbf{s}, \mathbf{s}' \in \R^d.
    \end{equation}
\end{assumption}
Recall that $\hat{u}_{k}(\mathbf{w})$ is the numerical value function at time $t_k$ computed by \cref{alg: sparse hermite polynomial expansions}, which approximates the exact value ${u}_{k}(\mathbf{w})$ for $k = 0,1,\dots,N-1$. We aim to establish an upper bound for
$\max_{0\le k\le N-1} \|u_k -\hat{u}_k\|_{{L}^2_{\omega_k}(\R^d)}^2$.

\subsection{One-step error estimation}
First, we analyze the one step error over the interval $[t_k,t_{k+1}]$ for $k = 0,1,\dots,N-1$, i.e., the numerical value function $\hat{u}_{k}(\mathbf{w})$ as an approximation to the exact one ${u}_{k}(\mathbf{w})$. Given the previous value function $\hat{u}_{k+1}$, we define the current CVF  $ \check{c}_k(\mathbf{W}_k)$ by
\begin{equation}\label{eq:apprx-conti}
    \check{c}_k(\mathbf{W}_k) = \mathbb{E}[\hat{u}_{k+1}(\mathbf{W}_{k+1}) | \mathbf{W}_k].
\end{equation}
The classic backward Euler scheme for BSDE approximates $\nabla\check{c}_k$ by
\begin{equation}\label{eq:grad-est}
     \check{Z}_k:=\frac{1}{\Delta t} \mathbb{E}[\hat{u}_{k+1}(\mathbf{W}_{k+1})\Delta \mathbf{W}_k | \mathbf{W}_k].
\end{equation}
With the sparse Hermite polynomial ansatz space $P_{I,k}$, let $c_k^* \in P_{I,k}$ be the best approximation to $\check{c}_k$ in $H^1_{\omega_k}(\mathbb{R}^d)$ defined by
\begin{equation}\label{eq:best-approx}
    c^*_k := \argmin_{\psi \in P_{I,k}} \|\check{c}_k - \psi\|_{H^1_{\omega_k}(\R^d)}.
\end{equation}
Also we define the best approximation error $\mathcal{E}_k^{\rm best}$ by
\begin{equation} \label{eq: defi best approx}
    \mathcal{E}_k^{\rm best} := \|\check{c}_k - c_k^*\|_{H^1_{\omega_k}(\R^d)}^2.
\end{equation}
Next, we denote the statistical error of solving the discrete least squares problem \cref{eq: empirical risk} by
\begin{align}\label{eq:stat-lsm}
\mathcal{E}^{\rm stat}_k := \|\hat{c}_k - c_k^{\op{CLS}}\|_{{L}^2_{\omega_k}(\R^d)}^2,
\end{align}
where $c_k^{\rm CLS}$ solves the continuous least squares problem \cref{eq: expected risk} for $k = 0,1,\dots,N-1$.
In \cref{thm: one-step error} below, we derive the convergence of $c_k^{\op{CLS}}$ to the exact $c_k$ provided that the previous approximation error $\|u_{k+1} - \hat{u}_{k+1}\|_{L_{\omega_{k+1}}^2(\R^d)}^2$,  $\mathcal{E}_k^{\rm best}$ in \cref{eq: defi best approx}, and time step size $\Delta t$ are small. We first provide in \cref{lem:best-error-HC} an estimate of $\mathcal{E}_k^{\rm best}$ in \cref{eq: defi best approx}, which is a direct application of \cite[Theorem 4.2]{luo2018error}.

\begin{lemma} \label{lem:best-error-HC}
    For the hyperbolic cross index set $I$ with maximum order $p$ defined in \cref{eq:HC-defi}, given $\check{c}_k \in \mathcal{K}^m_{\omega_k}(\R^d)$, we have
    \begin{equation*}
        \mathcal{E}_k^{\rm best} \le C(m, d)\frac{1}{p^{m-1}} |\check{c}_k|_{\mathcal{K}^m_{\omega_k}(\R^d)}^2,
    \end{equation*}
    where $\mathcal{K}^m_{\omega_k}(\R^d)$ is the weighted Koborov-type space defined by
    \begin{equation*}
        \mathcal{K}^m_{\omega_k}(\R^d) = \{u: \partial^{\bm{\alpha}} u \in {L}^2_{\omega_k}(\R^d), 0\le |\bm{\alpha}|_\infty \le m\},
    \end{equation*}
    and $|\cdot|_{\mathcal{K}^m_{\omega_k}(\R^d)}$ is the seminorm defined by
    \begin{equation*}
        |u|_{\mathcal{K}^m_{\omega_k}(\R^d)} = \left( \sum_{|\bm{\alpha}|_\infty = m} \|\partial^{\bm{\alpha}} u\|_{{L}^2_{\omega_k}(\R^d)}^2 \right)^{1/2}.
    \end{equation*}
\end{lemma}

The next result provides a one-step error estimate of the $Z$ component in the BSDE by the backward Euler scheme \cref{eq:grad-est}. The proof is inspired by the idea of \cite[pages 816-817]{gobet2007error}.
By martingale representation theorem, the definition of $\check{c}_k$ \cref{eq:apprx-conti} implies the existence of a square integrable process $\tilde{Z}_s$, $t_k\le s\le t_{k+1}$, such that
\begin{equation} \label{eq: martingale representation}
    \hat{u}_{k+1}(\mathbf{W}_{k+1}) = \check{c}_k(\mathbf{W}_k) + \int_{t_k}^{t_{k+1}} \tilde{Z}_t \cdot \dd \mathbf{W}_t.
\end{equation}
\begin{lemma} \label{lem:one-step Z}
Let $\check{Z}_k$ and $\tilde{Z}_t$ be defined in \cref{eq:grad-est} and \cref{eq: martingale representation}. Then for some constant $C_1>0$, \begin{equation*}
        \mathbb{E}\left[ (\check{Z}_k - \tilde{Z}_{t_k})^2 \right] \le C_1 \Delta t^2.
    \end{equation*}
\end{lemma}
\begin{proof}
By the integration by parts formula of Malliavin calculus, we obtain
    \begin{equation*}
        \check{Z}_k = \frac{1}{\Delta t} \mathbb{E}\left[ \hat{u}_{k+1}(\mathbf{W}_{k+1}) \Delta \mathbf{W}_k | \mathbf{W}_k \right] = \frac{1}{\Delta t} \mathbb{E}\left[ \int_{t_k}^{t_{k+1}} D_t\mathbf{W}_{k+1} \nabla \hat{u}_{k+1}(\mathbf{W}_{k+1}) \,\dd t \bigg| \mathbf{W}_k \right],
    \end{equation*}
    where $D_t \mathbf{W}_{k+1}$ is an identity matrix of size $d\times d$.
Together with the identity
    \begin{equation*}
        \tilde{Z}_{t_k} = \nabla \check{c}_k(\mathbf{W}_k) = \frac{1}{\Delta t}\int_{t_k}^{t_{k+1}} \nabla \check{c}_k(\mathbf{W}_k) \,\dd t,
    \end{equation*}
it further leads to
    \begin{align*}
        \check{Z}_k - \tilde{Z}_{t_k} &= \frac{1}{\Delta t}\mathbb{E}\left[ \int_{t_k}^{t_{k+1}} \nabla \hat{u}_{k+1}(\mathbf{W}_{k+1}) - \nabla \check{c}_k(\mathbf{W}_k) \,\dd s \bigg| \mathbf{W}_k \right].
    \end{align*}
Then, by \cite[Equation (26)]{gobet2007error}, we obtain
    \begin{equation*}
        \check{Z}_k - \tilde{Z}_{t_k} = \int_{t_k}^{t_{k+1}} \mathbb{E}\left[ G(t, \mathbf{W}_t) |\mathbf{W}_k \right] \,\dd t,
    \end{equation*}
    for a bounded function $G$. Hence, there holds $|\check{Z}_k - \tilde{Z}_{t_k}| = \mathcal{O}(\Delta t)$, and there exists a constant $C_1>0$ such that $\mathbb{E}\left[ (\check{Z}_k - \tilde{Z}_{t_k})^2) \right] \le C_1 \Delta t^2$.
\end{proof}

Next, we provide an error estimate of solving the continuous least squares problem \cref{eq: expected risk} in terms of the error of the previous value function, the best approximation error in the sparse Hermite polynomial ansatz space \cref{eq: defi best approx} and the time step size $\Delta t$. The proof is inspired by the foundational work \cite{hure2020deep}. We first derive a one-step error propagation by using Young's inequality and then apply discrete Gronwall's inequality to establish global convergence. In contrast to the approach in \cite{hure2020deep}, our analysis also incorporates the statistical error introduced by the Monte Carlo method and the one-step error proven in \cref{lem:one-step Z}. Moreover, the best approximation error of sparse Hermite polynomial approximation is given in \cref{lem:best-error-HC}.
\begin{theorem} \label{thm: one-step error}
For $\Delta t \le 1/4$, with the constant $C_1>0$ in \cref{lem:one-step Z}, there holds
    \begin{equation*}
        \|c_k - c_k^{\op{CLS}}\|_{{L}^2_{\omega_k}(\R^d)}^2 \le (1 + 4 \Delta t) \|u_{k+1} - \hat{u}_{k+1}\|_{{L}^2_{\omega_{k+1}}(\R^d)}^2 + \tfrac{1}{2\Delta t} \mathcal{E}_k^{\rm best} + C_1 \Delta t^2.
    \end{equation*}
\end{theorem}
\begin{proof}
By the inequality $(a+b)^2 \le (1+ 4 \Delta t) a^2 + (1+\frac{1}{4 \Delta t})b^2$, we have
\begin{align*}
    \|c_k - c_k^{\op{CLS}}\|_{{L}^2_{\omega_k}(\R^d)}^2 &\le (1+4 \Delta t) \|c_k - \check{c}_k\|_{{L}^2_{\omega_k}(\R^d)}^2 + (1+\tfrac{1}{4\Delta t}) \|\check{c}_k - c_k^{\rm CLS}\|_{{L}^2_{\omega_k}(\R^d)}^2 \\
    &\le (1+4 \Delta t) \|c_k - \check{c}_k\|_{{L}^2_{\omega_k}(\R^d)}^2 + \tfrac{1}{2\Delta t} \|\check{c}_k - c_k^{\op{CLS}}\|_{{L}^2_{\omega_k}(\R^d)}^2,
\end{align*}
since $\Delta t \le 1/4$. For the first term, the definitions \cref{eq: continuation value W} and \cref{eq:apprx-conti}, the tower property, and Jensen's inequality lead to
\begin{equation*}
\begin{aligned}
    \mathbb{E}\left[ \left( c_k(\mathbf{W}_k) - \check{c}_k(\mathbf{W}_k) \right)^2 \right] &= \mathbb{E}\left[ \left(\mathbb{E}\left[ u_{k+1}(\mathbf{W}_{k+1}) - \hat{u}_{k+1}(\mathbf{W}_{k+1})|\mathbf{W}_k \right] \right)^2 \right] \\
    &\le \mathbb{E}\left[ \left(u_{k+1}(\mathbf{W}_{k+1}) - \hat{u}_{k+1}(\mathbf{W}_{k+1}) \right)^2 \right] = \|u_{k+1} - \hat{u}_{k+1}\|_{{L}^2_{\omega_{k+1}}(\R^d)}^2.
\end{aligned}
\end{equation*}
Hence, it remains to show
\begin{equation} \label{eq:remains-to-show}
    \|\check{c}_k - c_k^{\op{CLS}}\|_{{L}^2_{ \omega_{k}}(\R^d)}^2 \le \mathcal{E}_k^{\rm best} + 2C_1\Delta t^3.
\end{equation}
By plugging \cref{eq: martingale representation} into the quadratic loss function $E_k(\cdot)$ in \cref{eq: expected risk}, we obtain
\begin{align*}
    E_k(\psi) &= \mathbb{E}\left[ \left( \check{c}_k(\mathbf{W}_k) - \psi(\mathbf{W}_k) + \int_{t_k}^{t_{k+1}} \tilde{Z}_s\cdot \dd \mathbf{W}(s) - \nabla \psi(\mathbf{W}_k)\cdot \Delta \mathbf{W}_k \right)^2 \right] \\
    &= \mathbb{E}\left[ \left( \check{c}_k(\mathbf{W}_k) - \psi(\mathbf{W}_k) \right)^2 \right] + \mathbb{E}\left[ \left( \int_{t_k}^{t_{k+1}} \tilde{Z}_s\cdot \dd \mathbf{W}(s) - \nabla \psi(\mathbf{W}_k)\cdot \Delta \mathbf{W}_k \right)^2 \right].
\end{align*}
Now It\^{o} isometry implies the identity
\begin{align*}
    & \mathbb{E}\left[ \left( \int_{t_k}^{t_{k+1}} \tilde{Z}_s\cdot \dd \mathbf{W}(s) - \nabla \psi(\mathbf{W}_k)\cdot \Delta \mathbf{W}_k \right)^2 \right] \\
    =& \mathbb{E}\left[ \left( \int_{t_k}^{t_{k+1}} \tilde{Z}_s\cdot \dd \mathbf{W}(s) - \int_{t_k}^{t_{k+1}} \check{Z}_k \cdot \dd \mathbf{W}(s) + \check{Z}_k\cdot \Delta \mathbf{W}_k - \nabla \psi(\mathbf{W}_k)\cdot \Delta \mathbf{W}_k \right)^2 \right]\\
    =& \mathbb{E}\left[ \int_{t_k}^{t_{k+1}} |\tilde{Z}_s - \check{Z}_k|^2\,\dd s \right] + \Delta t \mathbb{E}\left[ \left|\check{Z}_k - \nabla \psi(\mathbf{W}_k) \right|^2 \right] \\
    &+ 2\mathbb{E}\left[ \int_{t_k}^{t_{k+1}} (\tilde{Z}_s - \check{Z}_k)\,\dd s \right] \cdot \mathbb{E}[\check{Z}_k - \nabla \psi(\mathbf{W}_k)].
\end{align*}
The definitions of $\check{Z}_k$ and $\tilde{Z}_s$ in \cref{eq:grad-est} and \cref{eq: martingale representation}, together with It\^{o} isometry, yield
\begin{align*}
\check{Z}_k = \frac{1}{\Delta t} \mathbb{E}\left[ \int_{t_k}^{t_{k+1}} \tilde{Z}_s \,\dd s \bigg| \mathbf{W}_k \right],
\end{align*}
which implies
\begin{align*}
\mathbb{E}\left[ \int_{t_k}^{t_{k+1}} (\tilde{Z}_s - \check{Z}_k)\,\dd s \right] = 0.
\end{align*}
Hence, we can express the loss function $E_k(\psi)$ as
\begin{equation*}
    E_k(\psi) = \mathbb{E}\left[ \left( \check{c}_k(\mathbf{W}_k) - \psi(\mathbf{W}_k) \right)^2 \right] + \mathbb{E}\left[ \int_{t_k}^{t_{k+1}} |\tilde{Z}_s - \check{Z}_k|^2\,\dd s \right] + \Delta t \mathbb{E}\left[ \left|\check{Z}_k - \nabla \psi(\mathbf{W}_k) \right|^2 \right].
\end{equation*}
The equality \cref{eq: expected risk} implies that $E_k(c_k^{\rm CLS}) \le E_k(\psi)$ for all $\psi\in P_{I,k}$. By taking $\psi := c_k^*$ as defined in \cref{eq:best-approx}, we obtain
\begin{align*}
    &\quad \mathbb{E}\left[ \left( \check{c}_k(\mathbf{W}_k) - c_k^{\rm CLS}(\mathbf{W}_k) \right)^2 \right] + \Delta t \mathbb{E}\left[ \left(\check{Z}_k - \nabla c_k^{\rm CLS}(\mathbf{W}_k) \right)^2 \right] \\
    &\le \mathbb{E}\left[ \left( \check{c}_k(\mathbf{W}_k) - c_k^*(\mathbf{W}_k) \right)^2 \right] + \Delta t \mathbb{E}\left[ \left(\check{Z}_k - \nabla c_k^*(\mathbf{W}_k) \right)^2 \right].
\end{align*}
Note that $\tilde{Z}_{t_k} = \nabla \check{c}_k(\mathbf{W}_k)$. By subtracting and adding $\tilde{Z}_{t_k}$, we have
\begin{align*}
    &\quad \|\check{c}_k - c_k^{\op{CLS}}\|_{{L}^2_{\omega_{k}}(\R^d)}^2 \\
    &\le \|\check{c}_k - c_k^*\|_{{L}^2_{\omega_{k}}(\R^d)}^2 + 2\Delta t \mathbb{E}\left[ (\check{Z}_k - \tilde{Z}_{t_k})^2 \right] + 2\Delta t \mathbb{E}\left[ |\nabla \check{c}_k(\mathbf{W}_k) - \nabla c_k^*(\mathbf{W}_k)|^2 \right] \\
    &\le \|\check{c}_k - c_k^*\|_{H^1_{ \omega_{k}}(\R^d)}^2 + 2C_1\Delta t^3,
\end{align*}
where the last inequality follows from the condition $2\Delta t \le 1$ and \cref{lem:one-step Z}. This shows the estimate \cref{eq:remains-to-show}, and completes the proof of the theorem.
\end{proof}

Now, we can give the one-step error propagation of $\|u_k - \hat{u}_k\|_{{L}^2_{\omega_k}(\R^d)}^2$ given $\|u_{k+1} - \hat{u}_{k+1}\|_{{L}^2_{\omega_{k+1}}(\R^d)}^2$.
\begin{corollary} \label{cor: one-step error}
For $\Delta t\le 1/6$, there holds
    \begin{equation*}
        \|u_k - \hat{u}_k\|_{{L}^2_{\omega_k}(\R^d)}^2 \le (1+6 \Delta t) \|u_{k+1} - \hat{u}_{k+1}\|_{{L}^2_{\omega_{k+1}}(\R^d)}^2 + \left(1 + \frac{1}{2\Delta t}\right) \mathcal{E}_k^{\rm best} + \frac{6}{5}C_1\Delta t^2 + \frac{6}{5\Delta t} \mathcal{E}_k^{\rm stat},
    \end{equation*}
    where $\mathcal{E}_k^{\rm best}$ and $\mathcal{E}_k^{\rm stat}$ are defined in \cref{eq: defi best approx} and \cref{eq:stat-lsm}, and $C_1>0$ is the constant in \cref{lem:one-step Z}.
\end{corollary}
\begin{proof}
First, using inequality $(a+b)^2 \ge (1-\Delta t)a^2 - \frac{1}{\Delta t}b^2$ for any $a,b\in\mathbb{R}$, leads to
\begin{equation*}
    \|c_k - c_k^{\op{CLS}}\|_{{L}^2_{\omega_k}(\R^d)}^2 \ge (1-\Delta t) \|c_k - \hat{c}_k\|_{{L}^2_{\omega_k}(\R^d)}^2 - \frac{1}{\Delta t}\|\hat{c}_k - c_k^{CLS}\|_{{L}^2_{\omega_k}(\R^d)}^2.
\end{equation*}
This and \cref{thm: one-step error} yield
\begin{align*}
    &\quad \|c_k - \hat{c}_k\|_{{L}^2_{\omega_k}(\R^d)}^2 \\
    &\le \frac{1+4\Delta t}{1-\Delta t} \|u_{k+1} - \hat{u}_{k+1}\|_{{L}^2_{\omega_{k+1}}(\R^d)}^2 + \frac{1}{2\Delta t(1-\Delta t)} \mathcal{E}_k^{\rm best} + \frac{C_1\Delta t^2}{1-\Delta t} + \frac{1}{\Delta t(1-\Delta t)} \mathcal{E}_k^{\rm stat} \\
    &\le (1 + 6\Delta t) \|u_{k+1} - \hat{u}_{k+1}\|_{{L}^2_{\omega_{k+1}}(\R^d)}^2 + \left(1 + \frac{1}{2\Delta t} \right) \mathcal{E}_k^{\rm best} + \frac{6}{5}C_1\Delta t^2 + \frac{6}{5\Delta t} \mathcal{E}_k^{\rm stat},
\end{align*}
provided that $\Delta t\le 1/6$.
Finally, using inequality $|\max(a,b) - \max(a,c)| \le |b-c|$, we derive
\begin{equation*}
    \|u_k - \hat{u}_k\|_{{L}^2_{\omega_k}(\R^d)}^2 = \|\max(G_k, c_k) - \max(G_k, \hat{c}_k)\|_{{L}^2_{\omega_k}(\R^d)}^2 \le \|c_k - \hat{c}_k\|_{{L}^2_{\omega_k}(\R^d)}^2.
\end{equation*}
This completes the proof.
\end{proof}

\subsection{Global error estimation}
Finally, we prove a global error estimate.
\begin{theorem} \label{thm: global error}
For $\Delta t \le 1/6$, there holds
    \begin{equation*}
        \max_{0\le k\le N-1} \|u_k - \hat{u}_k\|_{{L}^2_{\omega_k}(\R^d)}^2 \le C\left( \Delta t +  N\sum_{k=0}^{N-1} \left( \mathcal{E}_k^{\rm best} + \mathcal{E}_k^{\rm stat} \right) \right),
    \end{equation*}
    where the constant $C = \max\left(\frac{6}{5T}e^{6T}, \frac{6T}{5}C_1 e^{6T} \right)$ only depends on the finite time horizon $T$ and the constant $C_1>0$ defined in \cref{lem:one-step Z}.
\end{theorem}
\begin{proof}
    For $\Delta t\le 1/6$, \cref{cor: one-step error} implies
    \begin{equation*}
        \|u_k - \hat{u}_k\|_{{L}^2_{\omega_k}(\R^d)}^2 \le (1+6 \Delta t) \|u_{k+1} - \hat{u}_{k+1}\|_{{L}^2_{\omega_{k+1}}(\R^d)}^2 + \frac{6}{5\Delta t} \left( \mathcal{E}_k^{\rm best} + \mathcal{E}_k^{\rm stat} \right) + \frac{6}{5}C_1\Delta t^2.
    \end{equation*}
    Using the discrete Gronwall's inequality and $\Delta t = T/N$, we obtain
    \begin{equation*}
        \max_{0\le k\le N-1} \|u_k - \hat{u}_k\|_{{L}^2_{\omega_k}(\R^d)}^2 \le e^{6T} \|u_N - \hat{u}_N\|_{{L}^2_{\omega_N}(\R^d)}^2 + e^{6T} \frac{6N}{5T} \sum_{k=0}^{N-1} \left( \mathcal{E}_k^{\rm best} + \mathcal{E}_k^{\rm stat} \right) + e^{6T} \frac{6T}{5}C_1 \Delta t.
    \end{equation*}
    Since the terminal condition $u_N$ is known, the result follows by $\|u_N - \hat{u}_N\|_{{L}_{\omega_N}^2(\R^d)} = 0$.
\end{proof}

\Cref{thm: global error} implies that the numerical value function $\hat{u}_{k}$ computed by \cref{alg: sparse hermite polynomial expansions} at each time $t_k$ approximates the exact one ${u}_{k}$ well if $\Delta t$, $\mathcal{E}_k^{\rm best}$, and $\mathcal{E}_k^{\rm stat}$ are small. Using \cref{lem:best-error-HC}, the best approximation error $\mathcal{E}_k^{\rm best}$ decays geometrically with the increase of the order $p$ of sparse Hermite polynomials. By the law of large numbers, the statistical error
$\mathcal{E}_k^{\rm stat}$ between the discrete and the continuous least squares will be small when using a large number of sample paths.

\section{Numerical examples}\label{sec:examples}
In this section, we present several examples of high-dimensional Bermudan option pricing to show the efficiency and accuracy of \cref{alg: sparse hermite polynomial expansions}. In practice, the proposed method is applicable to other problems, e.g., high-dimensional optimal stopping problems, provided that the underlying dynamics follow an It\^{o} process.

The codes for the numerical experiments
can be founded in the GitHub repository \url{https://github.com/jiefeiy/glsm-american}. The accuracy of the computed prices $\hat{v}_0(\mathbf{s}_0)$ and deltas $\nabla \hat{v}_0(\mathbf{s}_0)$ are measured by the relative errors defined by
\begin{equation*}
    \frac{|\hat{v}_0(\mathbf{s}_0) - v_0^\dagger|}{|v_0^\dagger|} \times 100\% \text{ and } \frac{\|\nabla \hat{v}_0(\mathbf{s}_0) - \Delta_0^\dagger\|}{\|\Delta_0^\dagger\|} \times 100\%,
\end{equation*}
respectively, where $v_0^\dagger$ and $\Delta_0^\dagger$ are exact price and delta at time $t=0$. Unless otherwise stated, the results are the average of 10 independent runs. The parameter settings of the examples are summarized in \cref{tab: parameter set}, which has been considered previously \cite{yang2023optionpricing,sirignano2018dgm, chen2021deep, becker2019deep,fang2011fourier}. The least squares problem is solved using the built-in MATLAB function \texttt{cgs}. The maximum order $p$ of the hyperbolic cross index set is one important hyper-parameter. In the experiments, we gradually increase the maximum order such that the pricing results do not change significantly, following the standard practice in sparse polynomial approximation \cite{MR4257879}. See Table \ref{tab:geobaskput_order} for the study on the pricing results with respect to the maximum order $p$. In practice, a higher-dimensional problem often needs a lower maximum order of hyperbolic cross index set, which might be related to the intrinsic low dimensional structure of high dimensional functions.

\begin{table}[htbp]
    \centering
    \caption{The parameters used in Examples 1-5. Examples 3 and 4 share the parameters except the volatility $\sigma_i$.}
    \label{tab: parameter set}
   {\footnotesize{
    \begin{tabular}{ll}
        \hline
        Example & Parameters \\
        \hline
        1. Geometric basket put & $K = 100, T = 0.25, r = 0.03, \delta_i = 0, \sigma_i = 0.2,\rho_{ij} = 0.5, N = 50$ \\
        2. Geometric basket call & $K = 100, T = 2, r = 0, \delta_i = 0.02, \sigma_i = 0.25, \rho_{ij} = 0.75, N=50$ \\
        3. Max-call with $d$ symmetric assets & $K = 100, T = 3, r = 0.05, \delta_i = 0.1, \sigma_i = 0.2, \rho_{ij} = 0, N=9$ \\
        4. Max-call with $d$ asymmetric assets & $
            \sigma_i = \begin{cases}
            0.08 + 0.32\times(i-1)/(d-1), \text{ if } d\le 5 \\
            0.1 + i/(2d), \text{ if } d> 5
        \end{cases}$ \\
        5. Put option under Heston model & \shortstack{$K=10, T=0.25, r=0.1, v_0 = 0.0625, \rho = 0.1,$ \\ $\kappa = 5, \theta = 0.16, \nu = 0.9, N=50$} \\
        \hline
    \end{tabular}  
  }}
\end{table}

\subsection{Example 1: Bermudan geometric basket put}
Geometric basket options are benchmark tests for high-dimensional option pricing problems, since they can be reduced to one-dimensional problems, and thus highly accurate prices are available. Indeed, the price of the $d$-dimensional problem equals that of the one-dimensional American option with initial price $\hat s_0$, volatility $\hat\sigma$, and dividend yield $\hat\delta$ given respectively by
\begin{equation*}
    \hat{s}_0 = \Big(\prod_{i=1}^d s_0^i\Big)^{1/d}, \quad \hat{\sigma} = \frac{1}{d}\sqrt{\sum_{i,j}\sigma_i \sigma_j \rho_{ij}}, \quad \hat{\delta} = \frac{1}{d}\sum_{i=1}^d \Big( \delta_i + \frac{\sigma_i^2}{2}\Big) - \frac{\hat{\sigma}^2}{2}.
\end{equation*} 

We consider the example of Bermudan geometric basket put from \cite{kovalov2007pricing, yang2023optionpricing}. The exact prices are computed by solving the reduced one-dimensional problem via a quadrature and interpolation-based method \cite{yang2023optionpricing} for Bermudan options. We present in \cref{tab: geobaskput glsm lsm} the computed option prices and their relative errors using G-LSM and LSM, with the same ansatz space for the CVF. The results show that G-LSM achieves higher accuracy than LSM in high dimensions: G-LSM has a relative error $0.55\%$ for $d=15$, which is almost ten times smaller than that by LSM. That is, by incorporating the gradient information, the accuracy of LSM can be substantially improved.

\begin{table}[htbp]
    \centering
    \caption{The price and relative errors of Bermudan geometric basket put computed by LSM and G-LSM using $M=100,000$ samples, with $s_0^i = 100$ and $p=10$.}
    \label{tab: geobaskput glsm lsm}
    {\footnotesize{
    \begin{tabular}{ccccccc}
        \hline
        $d$ & $N_b$ & LSM & error & G-LSM & error & $v_0^\dagger$ \\
        \hline
        $1$ & 11 & 3.6715 & 0.16\% & 3.6532 & 0.34\% & 3.6658 \\
        $2$ & 29 & 3.1886 & 0.18\% & 3.1791 & 0.12\% & 3.1831 \\
        $3$ & 56 & 3.0113 & 0.28\% & 2.9896 & 0.44\% & 3.0029 \\
        $5$ & 141 & 2.8700 & 0.70\% & 2.8381 & 0.41\% & 2.8499 \\
        $10$ & 581 & 2.8030 & 2.71\% & 2.7160 & 0.48\% & 2.7290 \\
        $15$ & 1446 & 2.8258 & 5.15\% & 2.6726 & 0.55\% & 2.6874 \\
        \hline
    \end{tabular}  
    }}
\end{table}

\Cref{tab:geobaskput_order} gives the computed option prices and their relative errors by G-LSM with different maximum polynomial orders $p$ for the $20$-dimensional geometric basket put. The relative error decays steadily as the order $p$ of Hermite polynomials increases, which agrees with \cref{thm: global error}.

\begin{table}[htbp]
    \centering
    \caption{The computed prices and their relative errors by G-LSM for $d=20$ with polynomial order $p$ and $M=100,000$ samples.  $v_0^\dagger = 2.6664$.}
    \label{tab:geobaskput_order}
    {\footnotesize{
    \begin{tabular}{ccccc}
        \hline
        $p$ & 2 & 5 & 10 & 12 \\
        \hline
        $\hat{v}_0(\mathbf{s}_0)$ & 2.6389 & 2.6414 & 2.6529 & 2.6634 \\
        error & 1.03\% & 0.94\% & 0.50\% & 0.11\% \\
        \hline 
    \end{tabular}
    }}
\end{table}

\subsection{Example 2: American geometric basket call}
Now we consider the example of American geometric basket call option from \cite{sirignano2018dgm, chen2021deep} to demonstrate that the proposed G-LSM can achieve the same level of accuracy as the DNN-based method \cite{chen2021deep}, and take $M = 720,000$ samples  as in \cite{chen2021deep}. The prices and deltas are given in \cref{tab: price geobaskcall} and \cref{tab: deltas geobaskcall}, respectively, where the results of \cite{chen2021deep} are the average of 9 independent runs. From \cref{tab: price geobaskcall}, both methods have similar accuracy for the price. From \cref{tab: deltas geobaskcall}, the relative error of delta using G-LSM and DNN varies slightly with the dimension $d$. This is probably because the DNN-based method computes the delta via a sample average, while G-LSM uses the derivatives of the value function directly. The relative error of the delta computed by G-LSM increases slightly for larger dimensions, possibly due to the small magnitude of the exact delta values.

\begin{table}[htbp]
    \centering
    \caption{Prices of American geometric basket call at $t=0$ using $M = 720,000$ samples.}
    \label{tab: price geobaskcall}
    {\footnotesize{
    \begin{tabular}{cccc cc c cc}
        \hline
         &  & & & \multicolumn{2}{c}{G-LSM} && \multicolumn{2}{c}{DNN \cite{chen2021deep}}   \\
         \cline{5-6} \cline{8-9}
        $d$ & $s_0^i$ & $p$ & exact & price & error && price & error  \\
        \hline
        7 & 100 & 10 & 10.2591 & 10.2475 & 0.11\% && 10.2468 & 0.12\%  \\
        13 & 100 & 10 & 10.0984 & 10.0781 & 0.20\% && 10.0822 & 0.16\%  \\
        20 & 100 & 10 & 10.0326 & 10.0141 & 0.18\% && 10.0116 & 0.21\%  \\
        100 & 100 & 6 & 9.9345 &  9.8980  & 0.37\%  && 9.9163 & 0.18\%  \\
        \hline 
    \end{tabular}
    }}
\end{table}

\begin{table}[htbp]
    \centering
    \caption{Deltas of American geometric basket call at $t=0$ using $M = 720,000$ samples. $\mathbf{1} = (1, 1, \dots, 1)^\top$.}
    \label{tab: deltas geobaskcall}
    {\footnotesize{
    \begin{tabular}{cccc cc c cc}
        \hline
         &  &  & & \multicolumn{2}{c}{G-LSM} && \multicolumn{2}{c}{DNN \cite{chen2021deep}}   \\
         \cline{5-6} \cline{8-9}
        $d$ & $s_0^i$ & $p$ & exact & delta & error && delta & error  \\
        \hline
        7 & 100 & 10 & $0.0722*\mathbf{1}$ & $(0.0724,0.0725,\dots,0.0724)^\top$ & 0.32\% && $0.0717*\mathbf{1}$ & 0.69\% \\
        13 & 100 & 10 & $0.0387*\mathbf{1}$ & $(0.0389,0.0387,\dots,0.0388)^\top$ & 0.39\% && $0.0384*\mathbf{1}$ & 0.78\% \\
        20 & 100 & 10 & $0.0251*\mathbf{1}$ & $(0.0253,0.0253, \dots,0.0254)^\top$ & 0.59\% && $0.0249*\mathbf{1}$ & 0.80\% \\
        100 & 100 & 6 & $0.00502*\mathbf{1}$ &  $(0.00501,0.00508, \dots, 0.00498)^\top$ & 1.45\% && $0.00498*\mathbf{1}$ & 0.80\% \\
        \hline
    \end{tabular} 
    }}
\end{table}

With a priori knowledge of the regularity of the CVF, we can approximate functions with polynomials in high dimensions. We briefly compare the complexity of the two approaches. Compared with the DNN approximator, G-LSM can involve fewer unknown parameters, and involves a simpler numerical task (solving least-squares problems versus minimizing nonconvex losses). Indeed, at each fixed time step, \cite{chen2021deep} suggests training the neural network with $L=7$ hidden layers and width $d+5$ in each layer, leading to more than $L(d+5)^2$ parameters. In contrast, the number of undetermined parameters in G-LSM is the number  $N_b$ of basis functions, which has a cardinality $\mathcal{O}(p(\ln p)^{d-1})$. For example, for $d=100$, the neural network approach involves more than $77175$ parameters, whereas G-LSM with $p=6$ involves only $N_b = 15451$ basis functions.

Next, \cref{fig: exercise boundary} shows the classification results of continued and exercised data using G-LSM and LSM with the number of simulated paths $M = 100,000$ in $d = 7$ or 20. Compared with the exact exercise boundary, G-LSM achieves better accuracy in determining the exercise boundary than LSM despite of using the same ansatz and number of paths. Thus, even with the right ansatz space, LSM might fail the task of finding exercise boundary in high dimensions using only a limited number of samples. Compared with \cite[Figure 5]{chen2021deep} and \cite[Figure 6]{prcing-hedging-rnn-2023}, \cref{fig: exercise boundary} domonstrates that G-LSM can detect an accurate exercise boundary with fewer number of paths than the DNN-based method.

\begin{figure}[htbp]
    \centering
    \subfloat[G-LSM \label{fig: grad_herm_d7}]{\includegraphics[width = .45\textwidth]{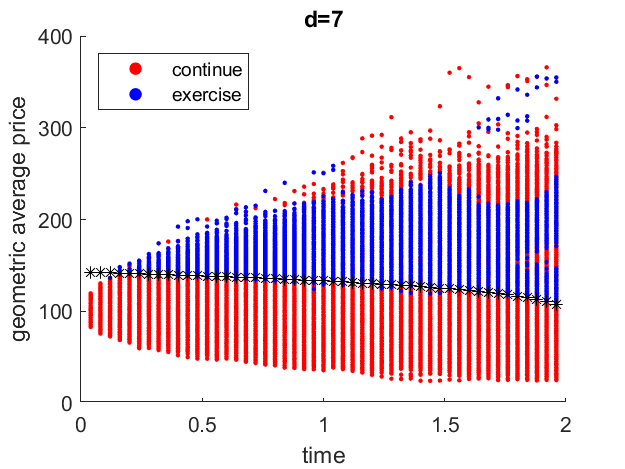}}
    \quad
    \subfloat[LSM \label{fig: lsmc_d7}]{\includegraphics[width = .45\textwidth]{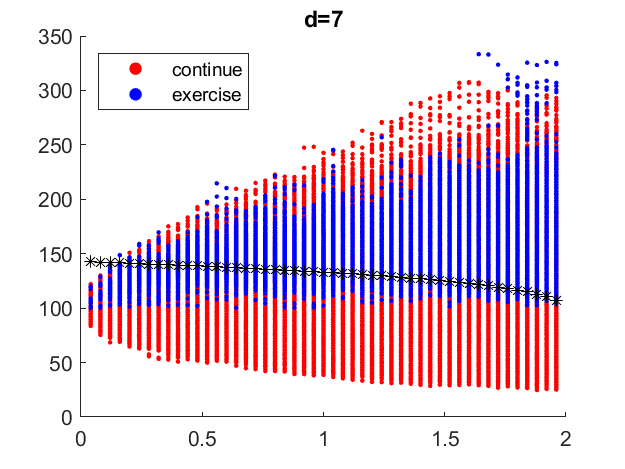}}
    \quad
    \subfloat[G-LSM \label{fig: grad_herm_d20}]{\includegraphics[width = .45\textwidth]{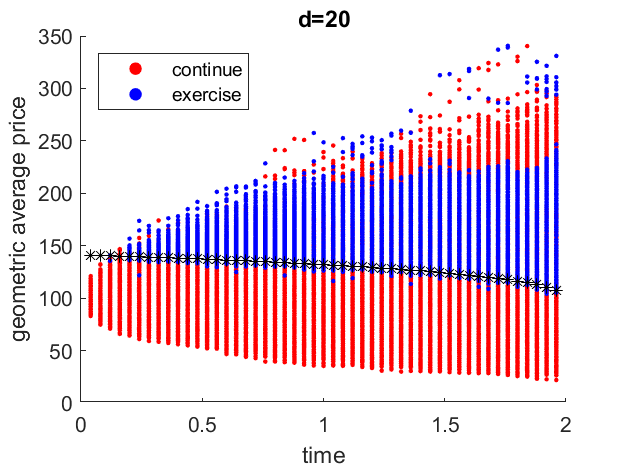}}
    \quad
    \subfloat[LSM \label{fig: lsmc_d20}]{\includegraphics[width = .45\textwidth]{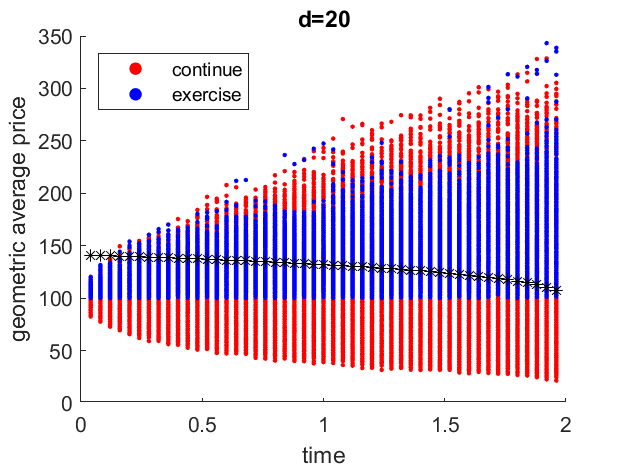}}
    \caption{Classification of the simulated continued and exercised data using $M=100,000$ samples, with $p=10$. Black star dots represent the exact exercise boundary.}
    \label{fig: exercise boundary}
\end{figure}

\subsection{Example 3: Bermudan max-call with symmetric assets} \label{subsec:ex3}
To benchmark G-LSM on high-dimensional problems without exact solutions and to validate the complexity analysis in \cref{sec:algorithm-cost}, we test Bermudan max-call option and report the computing time. The computing time is calculated as follows. For a fixed time step, $T_{\rm bas}$ is the time for generating basis matrix $\Phi$, $T_{\rm mat}$ is the time for assembling matrix $A$, $T_{\rm lin}$ is the time for solving linear system, and $T_{\rm up}$ is the time for updating values. The overall computing time is $T_{\rm tot} \approx (N-1)(T_{\rm bas} + T_{\rm mat} + T_{\rm lin} + T_{\rm up})$.

\Cref{tab: price max call} presents the prices and computing time  (in seconds) for Bermudan max-call options with $d$ symmetric assets. The reference 95\% confidence interval (CI) is taken from \cite{becker2019deep}. The reference CI is computed with more than $3000$ training steps and a batch of $8192$ paths in each step, which in total utilizes more than $10^7$ paths. The last five columns of the table report the computing time for the step 6, 7, 8, 9 in \cref{alg: sparse hermite polynomial expansions}, and the total time, respectively. All the computation for this example was performed on an Intel Core i9-10900 CPU 2.8 GHz desktop with 64GB DDR4 memory using  MATLAB R2023b. It is observed that the prices computed by G-LSM fall into or stay very close to the reference 95\% CI, confirming the high accuracy of G-LSM.  Furthermore, the time for generating basis matrix, $T_{\rm bas}$, dominates the overall computing time. Hence, the cost mainly arises from evaluating Hermite polynomials on sampling paths, which is also required by LSM. In comparison with LSM , $T_{\rm mat}$ is the extra cost to incorporate the gradient information and takes only a small fraction of the total time. Therefore, G-LSM  has nearly identical cost with LSM.

A method known as deep empirical risk minimization (deep ERM) has been introduced in \cite{reppen2024neural}. This approach involves parameterizing the stopping boundary using a neural network and includes Example 3 in \cite[Table 4]{reppen2024neural}. A comparison between \cref{tab: price max call} and \cite[Table 4]{reppen2024neural} reveals that both methods are effective in handling high dimensions. The runtime of the G-LSM method is shorter than that of deep ERM for moderately high dimensions (e.g., $d=2,3,5,10$), but longer for higher dimensions (e.g., $d=20,50$). This difference can be attributed to the fact that deep neural network-based approaches are generally less sensitive to dimensions but require a larger number of paths for training. For instance, the results presented in \cite[Table 4]{reppen2024neural} were obtained using $1,536,000$ simulation paths, whereas \cref{tab: price max call} utilized only $100,000$ paths.

We present in \cref{tab: time max call} the computing time for each basis function to verify that the complexity is almost linear in $N_b$ as analyzed in \cref{sec:algorithm-cost}. The ratio does not vary much with the dimension $d$. Since $N_b = \mathcal{O}(p (\ln p)^{d-1})$ with polynomials up to order $p$, the total computing cost of G-LSM exhibits a polynomial growth, which overcomes the so-called curse of dimensionality. \Cref{fig:hc-index-num} shows the growth of $N_b$ with respect to the dimension for $1\le d\le 200$ and $p=6$, indicating that $N_b$ exhibits a nearly quadratic growth in dimensions for $d\le 200$.

\begin{table}[htbp]
    \centering
    \caption{The results for Bermudan max-call options with $d$ symmetric assets using $M = 100,000$ samples. $s_0^i = 100$ for $i=1,\dots, d$.}
    \label{tab: price max call}
    {\footnotesize{
    \begin{tabular}{ccccccccccc}
        \hline
        $d$ &$p$& $N_b$ & reference 95\% CI & G-LSM & $T_{\rm bas}$ & $T_{\rm mat}$ & $T_{\rm lin}$ & $T_{\rm up}$ & $T_{\rm tot}$  \\
        \hline
        2 & 10 & 29 & $[13.880, 13.910]$ & 13.8970 & 0.1678  &  0.0128  &  0.0024  &  0.0039  &  1.5725 \\
        3 & 10 & 56 & $[18.673, 18.699]$ & 18.6715  & 0.2890  &  0.0265  &  0.0048  &  0.0060  &  2.6347 \\
        5 & 10 & 141 & $[26.138, 26.174]$ & 26.0553 & 0.5718  &  0.0758  &  0.0149  &  0.0132  &  5.4662 \\
        10& 10 & 581 & $ [38.300, 38.367]$ & 38.1738 & 1.9932  &  0.3164  &  0.1397  &  0.0497 &  19.9308 \\
        20& 10 & 2861 & $[51.549, 51.803]$ & 51.6508 & 11.3278  &  1.7153  &  3.2000  &  0.2434 & 131.9362 \\
        30& 5 & 1456 & $[59.476, 59.872]$ & 59.5475 & 7.8154  &  0.7750  &  0.7044  &  0.1247  & 75.3066 \\
        50& 5 & 3926 & $[69.560, 69.945]$ & 69.7216 & 29.3088  &  2.0516  &  5.6871  &  0.3281  & 299.4528 \\
        100& 4 & 5351 & $[83.357, 83.862]$ & 83.6777 & 71.7678  &  2.7185 &  10.8018  &  0.4306 & 684.9927 \\
        \hline
    \end{tabular}
    }}
\end{table}

\begin{table}[htbp]
    \centering
    \caption{Ratio of the computing time and number $N_b$ of basis functions. }
    \label{tab: time max call}
    {\footnotesize{
    \begin{tabular}{cccccc}
        \hline
        $d$ & $T_{\rm bas}/N_b$ & $T_{\rm mat}/N_b$ & $T_{\rm lin}/N_b$ & $T_{\rm up}/N_b$ & $T_{\rm tot}/N_b$  \\
        \hline
        2 & 0.0058  &  0.0004  &  0.0001  &  0.0001  &  0.0542 \\
        3 & 0.0052  &  0.0005  &  0.0001 &   0.0001  &  0.0470 \\
        5 & 0.0041  &  0.0005  &  0.0001  &  0.0001  &  0.0388 \\
        10 & 0.0034  &  0.0005  &  0.0002  &  0.0001  &  0.0343 \\
        20 & 0.0040  &  0.0006  &  0.0011  &  0.0001  &  0.0461 \\
        30 & 0.0054  &  0.0005  &  0.0005  &  0.0001  &  0.0517 \\
        50 & 0.0075  &  0.0005  &  0.0014  &  0.0001  &  0.0763 \\
        100 & 0.0134  &  0.0005  &  0.0020  &  0.0001  &  0.1280 \\
        \hline 
    \end{tabular}
    }}
\end{table}

\begin{figure}[htbp]
    \centering
    \includegraphics[width = 0.45\linewidth]{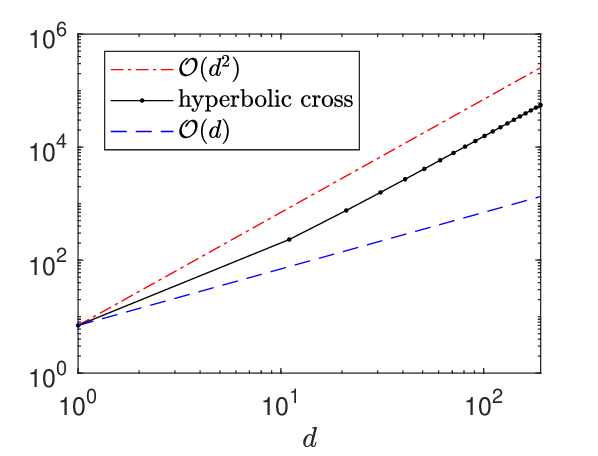}
    \caption{The number  $N_b$ of basis functions in dimension $1\le d \le 200$ with hyperbolic cross index set versus linear and quadratic scale in $d$, with $p=6$. }
    \label{fig:hc-index-num}
\end{figure}

We mention two bottlenecks of the proposed G-LSM for solving higher-dimensional problems.
The first is the main memory requirement of the method, i.e., storing the matrix $A\in \mathbb{R}^{M\times N_b}$, which takes $M\times N_b\times 8$ bytes in double-precision floating-point format. For $d=100$, the memory for storing $A$ is about 4GB. This limits the application of G-LSM with direct solvers in higher dimensions. To remedy the issue, one can solve the linear system on a large RAM server, or use single-precision floating-point, or with stochastic gradient descent. Second, although hyperbolic cross type sparse index can overcome the curse of dimensionality in the sense that the computational complexity does not grow exponentially with respect to $d$, cf. \cref{fig:hc-index-num}, the practical implementation for large $d$ can still be non-trivial even when the computational complexity grows only polynomially with respect to $d$. Thus, we have presented the numerical examples only up to $d=100$.

\Cref{fig:glsm_lsm_maxcall_classification} shows the classification of continued and exercised sample points computed by G-LSM and LSM in the example of two-dimensional max-call. G-LSM yields a smoother exercise boundary than LSM. Compared with the exercise boundary computed in literature \cite[Figure 3]{reppen2023deep}, G-LSM exhibits higher accuracy, albeit that the same ansatz space for the CVF is employed. 

\begin{figure}[htbp]
    \centering
    \subfloat[G-LSM, at $t_1$\label{fig: glsm_maxcall_t_1}]{\includegraphics[width = .32\textwidth]{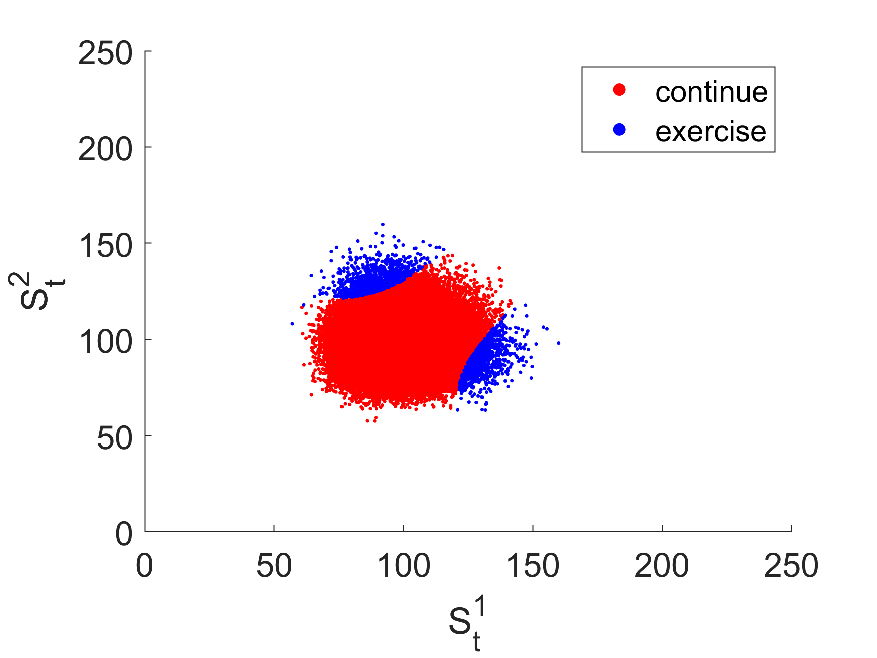}}
    \subfloat[G-LSM, at $t_4$ \label{fig: glsm_maxcall_t_4}]{\includegraphics[width = .32\textwidth]{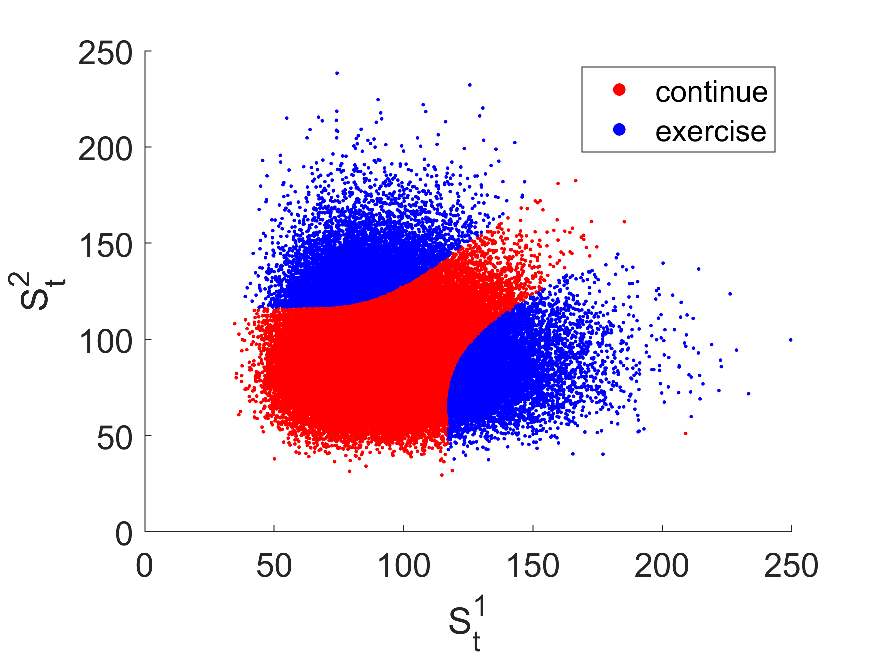}}
    \subfloat[G-LSM, at $t_7$ \label{fig: glsm_maxcall_t_7}]{\includegraphics[width = .32\textwidth]{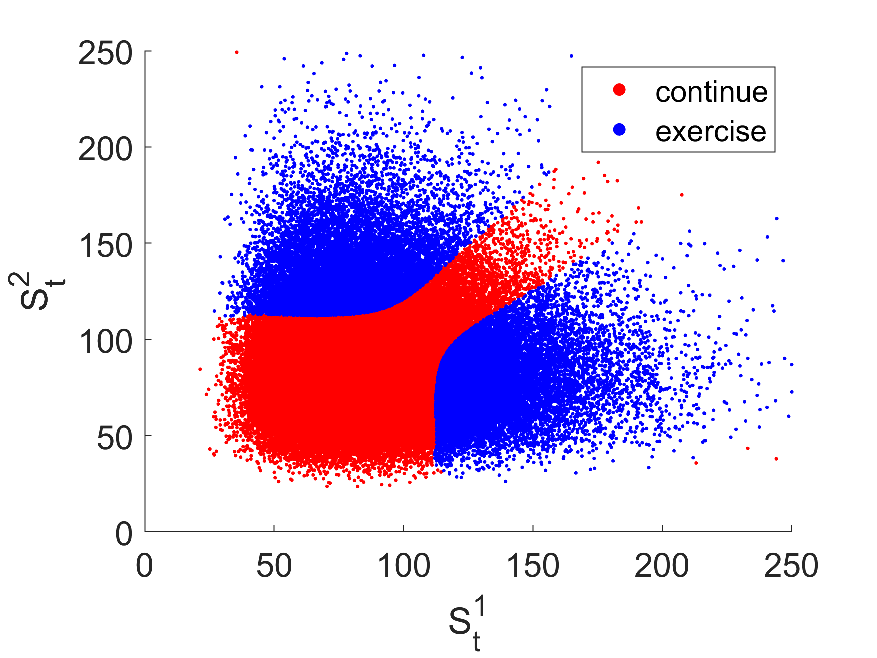}}
    \quad
    \subfloat[LSM, at $t_1$\label{fig: lsm_maxcall_t_1}]{\includegraphics[width = .32\textwidth]{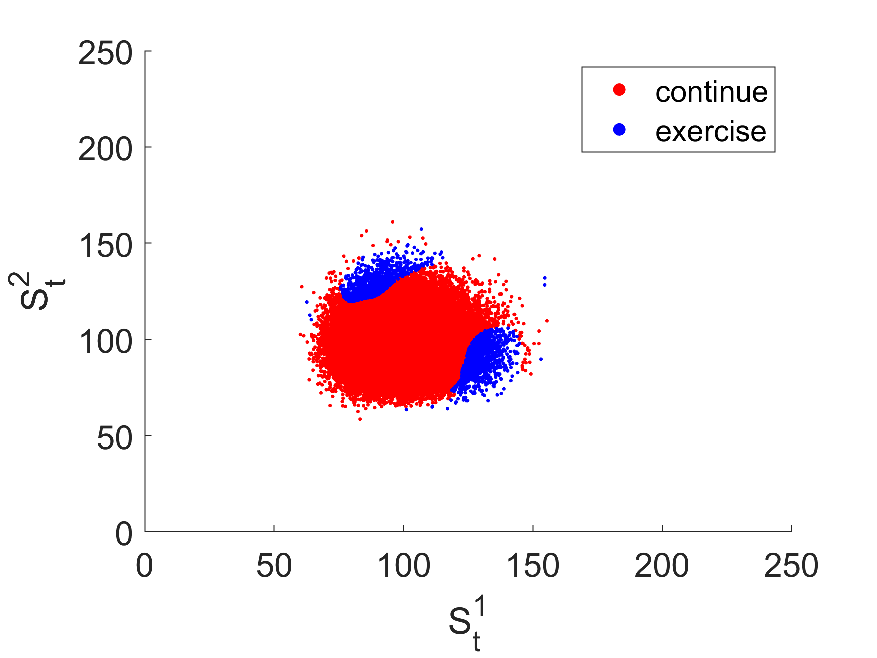}}
    \subfloat[LSM, at $t_4$ \label{fig: lsm_maxcall_t_4}]{\includegraphics[width = .32\textwidth]{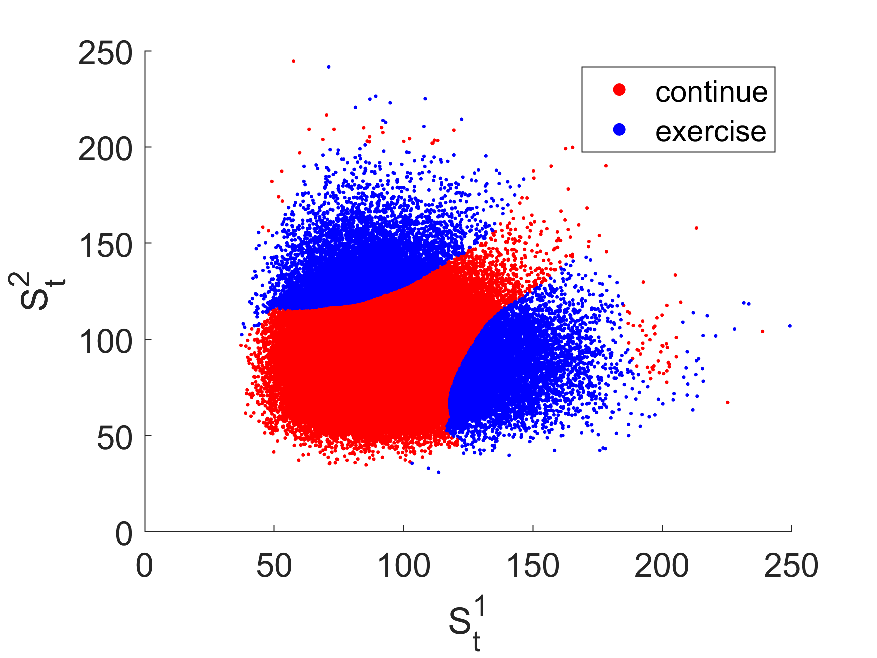}}
    \subfloat[LSM, at $t_7$ \label{fig: lsm_maxcall_t_7}]{\includegraphics[width = .32\textwidth]{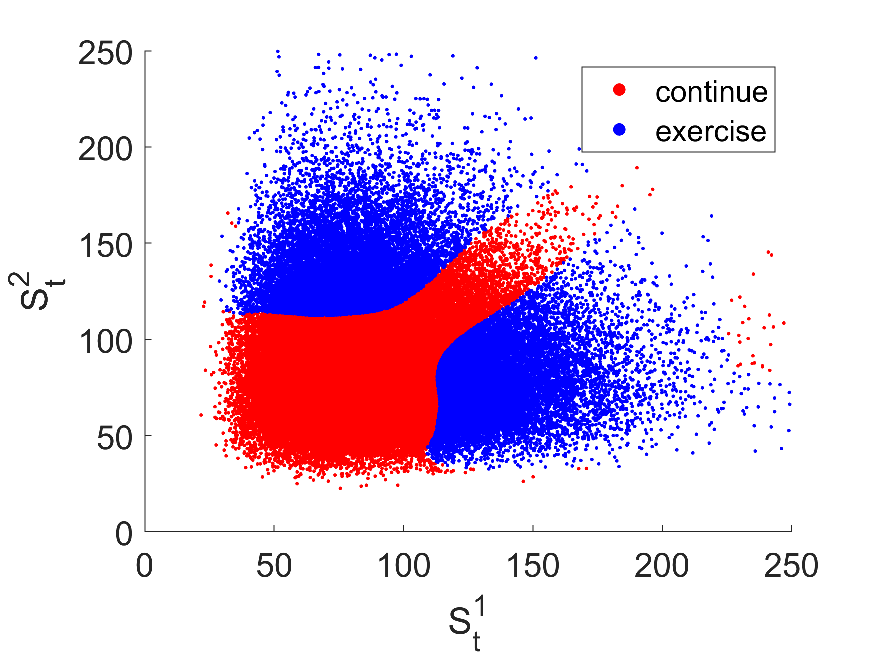}}
    \caption{Classification of simulated continued and exercised data for 2-$d$ Bermudan max-call option at different times $t_k$, $k=1, 4, 7$. Here, we take $s_0^i = 100$ for $i=1,2$, $M=100,000$ and $p=20$.}
    \label{fig:glsm_lsm_maxcall_classification}
\end{figure}

\subsection{Example 4: Bermudan max-call with asymmetric assets}
Example 3 assumes that all underlying assets follow the same dynamic. To further show the robustness of G-LSM, we consider the Bermudan max-call option but each asset has different volatility. The reference 95\% confidence interval (CI) is taken from \cite{becker2019deep}. The pricing results for different initial price $s_0^i$ and dimension $d$ are listed in \cref{tab: price max call asym}. The standard error (s.e.) is calculated by $\sqrt{\frac{1}{10(10-1)}\sum_{i=1}^{10}(v_0^{(i)} - \bar{v})^2}$ with $\bar{v}$ being the average of $10$ independent runs. Similar to Example 3, the prices computed by G-LSM always fall into or stay very close to the reference 95\% CI.

\begin{table}[htbp]
    \centering
    \caption{Results for Bermudan max-call options with $d$ asymmetric assets, with $M = 100,000$. }
    \label{tab: price max call asym}
    {\footnotesize{
    \begin{tabular}{cccccc}
        \hline
        $d$ & $p$ & $s_0^i$ & reference 95\% CI & G-LSM & s.e. \\
        \hline
        2 & 10 & 90 &  $[14.299, 14.367]$ & 14.3472 & 0.0230 \\
        2 & 10 & 100 & $[19.772, 19.829]$ & 19.8019 & 0.0371 \\
        2 & 10 & 110 & $ [27.138, 27.163]$ & 27.1041 & 0.0213 \\
        3 & 10 & 90 & $ [19.065, 19.104]$ & 19.0266 & 0.0241 \\
        3 & 10 & 100 & $[26.648, 26.701]$ & 26.6931 & 0.0411 \\
        3 & 10 & 110 & $ [35.806, 35.835]$ & 35.8363 & 0.0472 \\
        5 & 10 & 90 & $[27.630, 27.680]$ & 27.6032  &  0.0259 \\
        5 & 10 & 100 & $[37.940, 38.014]$ & 37.9309  &  0.0405 \\
        5 & 10 & 110 & $ [49.445, 49.533]$ & 49.3711  &  0.0473 \\
        10 & 10 & 90 & $ [85.857, 86.087]$ & 85.8221  &  0.0376 \\
        10 & 10 & 100 & $[104.603, 104.864]$ & 104.7052  &  0.1029 \\
        10 & 10 & 110 & $ [123.570, 123.904]$ & 123.4777  &  0.0686 \\
        20 & 10 & 90 & $ [125.819, 126.383]$ & 126.4276  &  0.0980 \\
        20 & 10 & 100 & $[149.480, 150.053]$ & 150.4028  &  0.1194 \\
        20 & 10 & 110 & $[173.144, 173.937]$ & 173.8584  &  0.1349 \\
        30 & 5 & 90 & $[154.378, 155.039]$ & 154.6913  &  0.1128 \\
        30 & 5 & 100 & $[181.155, 182.033]$ & 181.6733  &  0.1385 \\
        30 & 5 & 110 & $ [208.091, 209.086]$ & 208.1267  &  0.1160 \\
        50 & 5 & 90 & $[195.793, 196.963]$ & 196.6921  &  0.0890 \\
        50 & 5 & 100 & $[227.247, 228.605]$ & 227.7831  &  0.1385  \\
        50 & 5 & 110 & $ [258.661, 260.092]$ & 259.7261  &  0.1413 \\
        100 & 4 & 90 & $ [263.043, 264.425]$ & 263.0543  &  0.1515 \\
        100 & 4 & 100 & $[301.924, 303.843]$ & 302.0623  &  0.2130  \\
        100 & 4 & 110 & $ [340.580, 342.781]$ & 340.8581  &  0.2233 \\
        \hline 
    \end{tabular}
    }}
\end{table}

\subsection{Example 5: Bermudan put under Heston model} \label{subsec:heston}
The previous experiments and theoretical analysis are concerned with the most frequently used multi-asset Black-Scholes model, cf. \cref{sec:multi model}. We now generalize G-LSM to price Bermudan option under the Heston model. The Heston model defines the dynamic of the log-price process, $X_t^1 := \ln(S_t/S_0)$, and the volatility process, $v_t$, by two-dimensional SDEs:
\begin{align*}
    \dd X_t^1 &= (r-\tfrac{1}{2}v_t)\,\dd t + \sqrt{v_t} \left(\rho \,\dd W_t^1 + \sqrt{1-\rho^2}\,\dd W_t^2 \right), \\
    \dd v_t &= \kappa(\theta - v_t)\,\dd t + \nu \sqrt{v_t}\,\dd W_t^1,
\end{align*}
where $W_t^1$ and $W_t^2$ are two independent Wiener processes, and the model parameters $r$, $\kappa$, $\theta$, $\nu$ and $\rho$ represent the interest rate, the speed of mean reversion, the mean level of variance, the variance of volatility process, and the correlation coefficient, respectively. Since the transition density of log-variance has better regularity \cite{fang2011fourier}, we take the log-variance process $X_t^2 := \ln(v_t)$ as the regression variable. Let $\mathbf{X}_t = [X_t^1, X_t^2]^\top$ be a two-dimensional process. The discounted continuation value at time $t_k$ is given by
\begin{equation*}
    c_k(\mathbf{X}_{t_k}) = \mathbb{E}[u_{k+1}(\mathbf{X}_{t_{k+1}}) | \mathbf{X}_{t_k}],
\end{equation*}
with $u_{k+1}$ being the discounted value function. In view of the martingale representation theorem and backward Euler approximation, we obtain
\begin{equation*}
    u_{k+1}(\mathbf{X}_{t_{k+1}}) \approx c_k(\mathbf{X}_{t_k}) + \big( \sigma(\mathbf{X}_{t_k})^\top \nabla c_k(\mathbf{X}_{t_k}) \big) \cdot \Delta \mathbf{W}_k,
\end{equation*}
where $\sigma(\cdot)$ is the covariance matrix of $\mathbf{X}_t$ given by
\begin{equation*}
    \sigma(\mathbf{X}_t) = \begin{bmatrix}
        \rho \exp(X_t^2/2) & \sqrt{1-\rho^2} \exp(X_t^2/2) \\
        \nu \exp(-X_t^2/2) & 0
    \end{bmatrix}.
\end{equation*}
The reference prices are computed by the COS method \cite{fang2011fourier}, with $2^7$ cosine basis functions to approximate the transition density and $2^7$ points for the quadrature rule in log-variance dimension. The results of G-LSM are calculated using 70 basis functions (polynomials up to order $20$) with Hermite polynomials in log-price dimension and Chebyshev polynomials in log-variance dimension. The results in \cref{tab:heston} show that the computed prices by G-LSM coincide with the reference prices up to the two places after the decimal separator for most cases. In \cref{fig:cos_glsm}, we plot the classification results of continued and exercised points using two approaches. Given that the G-LSM method is simulation-based, it has great potential for stochastic volatility models in higher dimensions.

\begin{table}[htbp]
    \centering
    \caption{Prices of Bermudan put option under the Heston model. $M=100,000$. $p=20$.}
    \label{tab:heston}
    {\footnotesize{
    \begin{tabular}{cccccc}
        \hline
        $s_0$ & 8 & 9 & 10 & 11 & 12 \\
        \hline
        COS & 1.9958 & 1.1061 & 0.5186 & 0.2131 & 0.0818 \\
        G-LSM & 1.9949 & 1.0972 & 0.5146 & 0.2118 & 0.0807 \\
        \hline 
    \end{tabular}
    }}
\end{table}

\begin{figure}[hbt!]
    \centering
    \subfloat[COS \label{fig: fig_cos_t25_So8}]{\includegraphics[width = .45\textwidth]{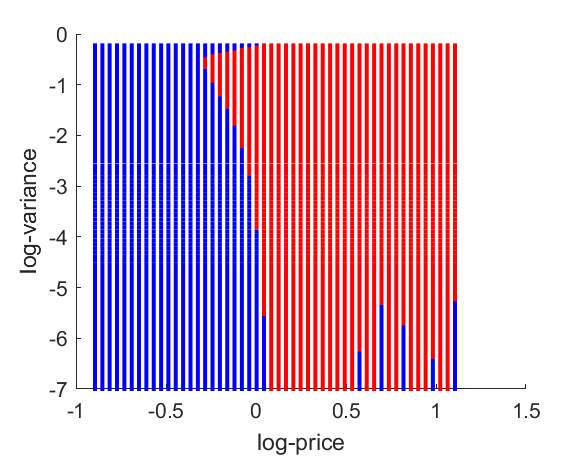}}
    \quad
    \subfloat[G-LSM \label{fig: fig_glsm_t25_So8}]{\includegraphics[width = .45\textwidth]{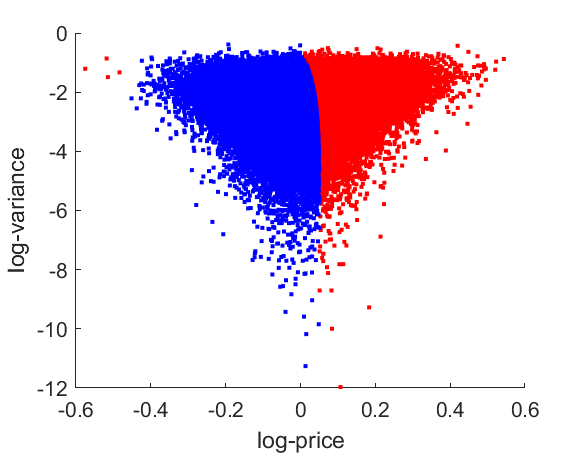}}
    \caption{Classification of continued and exercised grid/simulated points using COS and G-LSM under the Heston model at time $t_{25}$ with initial price $s_0 = 8$.}
    \label{fig:cos_glsm}
\end{figure}

\section{Conclusions and outlook} \label{sec:conclusion}
In this work, we have proposed a novel gradient-enhanced least squares Monte Carlo (G-LSM) method that employs sparse Hermite polynomials as the ansatz space to price and hedge American options. The method enjoys low complexity for the gradient evaluation, ease of implementation and high accuracy for high-dimensional problems. We analyzed rigorously the convergence of G-LSM  based on the BSDE technique, stochastic and Malliavin calculus. Extensive benchmark tests clearly show that it outperforms least squares Monte Carlo (LSM) in high dimensions with almost the same cost and it can also achieve competitive accuracy relative to the deep neural networks-based methods. 

There are several avenues for further research. The superiority of G-LSM over LSM in high dimensions indicates that matching option values at $t_{k+1}$ might be a better choice than at $t_k$ for approximating the continuation value function. There are other variants of LSM with different ansatz spaces, and it is natural to ask whether incorporating the gradient information will also result in improved performance for these variants. Moreover, to solve higher dimensional problems, e.g., $d=1000$, the hierarchical tensor train technique can be applied, which has been combined with LSM in \cite{bayer2023pricing}, and it is natural to combine the technique with G-LSM. Numerical results also show the potential of G-LSM for stochastic volatility models. It would be interesting to investigate both numerically and theoretically the proposed G-LSM for even more challenging financial models, e.g., rough volatility model. Numerically, for such non-Markovian models, the algorithm would require suitable Markovian approximation, and theoretically, one key aspect is to establish a suitable framework for analyzing the approximation error. Furthermore, it is of interest to assess the proposed G-LSM on practical applications, by a suitable combination with techniques for calibrating model parameters from real market data.

\section*{Acknowledgments}
The authors acknowledge the support of research computing facilities offered by Information Technology Services, the University of Hong Kong. The authors are also grateful to two anonymous referees for their constructive comments which have led a significant improvement in the quality of the paper.

\bibliographystyle{siamplain}
\bibliography{references}
\end{document}